\documentclass[aps,pra,showpacs,twoside, twocolumn, longbibleography, 10pt]{revtex4-2}
\usepackage[colorlinks=true, citecolor=red, urlcolor=blue ]{hyperref}
\usepackage{times,epsfig,amssymb,amsfonts,amsmath,bm,subfigure,mathtools,amsthm,braket,soul,enumitem,color,physics,graphics,graphicx}
\usepackage{mathrsfs}
\usepackage[normalem]{ulem}
\usepackage{comment}
\usepackage{epstopdf}
\newtheorem{theorem}{Theorem}
\newtheorem{corollary}{Corollary}

\newtheorem{lemma}{Lemma}

\newtheorem{observation}{Observation}
\newtheorem{protocol}{Protocol}

\def\>{\rangle}
\def\<{\langle}

\newcommand{\map}[1]{\mathcal{#1}}
\newcommand{\ga}[1]{{\color{blue}#1}}

\begin{document}

\title{
Random access Bell game by sequentially measuring the control of the quantum SWITCH
}

\author{Gaurang Agrawal$^1$, Saptarshi Roy$^2$}
\affiliation{$^1$Indian Institute of Science Education and Research, Homi Bhabha Rd, Pashan, Pune 411 008, India}
\affiliation{$^2$
QICI Quantum Information and Computation Initiative, Department of Computer Science, The University of Hong Kong, Pokfulam Road, Hong Kong	
	}

\begin{abstract}

Preserving quantum correlations such as Bell nonlocality in noisy environments remains a fundamental challenge for quantum technologies.
We introduce the Random Access Bell Game (RABG), a task where an entangled particle propagates through a sequence of identical noisy blocks, and the ability to violate a Bell inequality is tested at a randomly chosen point (access node). We consider a scenario where each noisy block is composed of two complete erasure channels, an extreme entanglement-breaking channel with vanishing quantum and classical capacities. We investigate the performance of the Random Access Bell Game in this configuration and attempt to mitigate the effect of noise by coherently controlling the order of each channel in the noise using the quantum {\tt SWITCH}. However, the quantum {\tt SWITCH} in its canonical setup with a coherent state in the control fails to provide any advantage in the Random Access Bell Game. Our main contribution is a protocol that leverages initial entanglement between the target and control of the quantum {\tt SWITCH} and employs sequential, unsharp measurements on the control system, showing that it is possible to guarantee a Bell violation after an arbitrarily large number of channel applications. Furthermore, our protocol allows for a near-maximal (Tsirelson bound) Bell violation to be achieved at any desired round, while still ensuring violations in all preceding rounds. We prove that this advantage is specific to generalized Greenberger-Horne-Zeilinger (GHZ) states, as the protocol fails for W-class states, thus providing an operational way to distinguish between these two fundamental classes of multipartite entanglement.

%We investigate this game in a setting where the noisy channel is constructed from two distinct entanglement-breaking maps, whose order of application can be controlled. We demonstrate that conventional probabilistic control over the channel order fails to preserve nonlocality. We then propose a protocol that utilizes the quantum {\tt SWITCH} to place the channel order under coherent quantum control. By employing sequential, unsharp measurements on the control system, we show that it is possible to guarantee a Bell violation after an arbitrarily large number of channel applications. Furthermore, our protocol allows for a near-maximal (Tsirelson bound) Bell violation to be achieved at any desired round, while still ensuring violations in all preceding rounds. This advantage is shown to be specific to Greenberger-Horne-Zeilinger (GHZ)-type entanglement, as the protocol fails for W-class states. Our results reveal a powerful method for protecting quantum nonlocality from repeated interactions with entanglement-breaking environments through the combined use of indefinite causal order and unsharp measurements.
\end{abstract}	

\maketitle

\section{Introduction}
\label{sec:intro}
\iffalse
and serving as a critical resource for quantum communication, cryptography, and computation .   \ga{Add more refs }.

\ga{
\begin{itemize}
\item existing literature in sequential
\item W vs GHZ states operational difference
\item  $\alpha$ in terms of concurrence shows the relation with entanglement measure
\item comment on partial erasure channels: more rounds with more noise
\end{itemize} 
}
\fi

Bell nonlocality stands as a cornerstone of quantum mechanics, revealing a stark departure from the classical principles of local realism \cite{Bell1964on, Clauser1969Oct, Brunner2014apr}. Violation of local realism has conceptually altered the understanding of quantum mechanics and the nature of quantum correlations in general \cite{Cirelson1980, navascas2007, Barizien2025, Le2025}. 
It has been demonstrated extensively in experiments \cite{bexp1, bexp2, bexp3} with recent breakthroughs in loophole-free Bell tests as well \cite{Hensen2015, Giustina2015, loop}.
Apart from foundational interests, Bell inequality violations turn out to be the crucial ingredient for various quantum information processing tasks. A non-exhaustive list includes tasks that range from device-independent cryptography \cite{diqkd1, diqkd2, diqkd3, diqkd4}, random number generation \cite{Pironio2010, Tan2016}, to reducing the communication complexity of certain problems \cite{comcomplexity1, comcomplexity2, comcomplexity3}. All these have motivated investigations of Bell inequality violation in physical platforms like continuous variable and quantum optical systems \cite{Walls1994, Chen2002, roy2018, Steinacker2025} and many-body systems \cite{Wang2002, Sadhukhan2015, Lee2022, Getelina2018} to make predictions that can be compared with actual experiments.
Another line of fundamentally important research involves the investigation of the relationship between Bell inequality violation and other facets of nonclassicality, like entanglement \cite{Popescu2014Bell,Buscemi2014All,Soorya2019Universality}. However, the quantum correlations that enable nonclassical effects are notoriously fragile, readily degrading through interaction with a noisy environment \cite{error1, error2}. Understanding and mitigating \cite{errormitigation1, em2} the impact of noise is therefore paramount for the development of robust quantum technologies \cite{Becher2023}. 

In this paper, we address this challenge by conceptualizing the task of a ``Random Access Bell Game." In this task, an entangled state is shared between two parties, Alice and Bob. Bob's particle traverses a noisy environment, modelled as a series of discrete interactions with identical noisy quantum channels. A referee can, at any point, randomly request that Alice and Bob perform a Bell test on their shared state. The central question is to determine the maximum number of noisy interactions, $k_{\max}$, the system can endure while retaining the ability to violate a Bell inequality, regardless of which ``access node" $k \leq k_{\max}$ the referee chooses. See Fig. \ref{fig:enter-label} for a schematic of the game.
For many types of noise, particularly entanglement-breaking channels \cite{Horodecki2003, Ruskai2003}, the ability to win this game is severely limited, with $k_{\max}$ often being zero. Specifically, we consider an extreme scenario where the noise is 
characterized by complete erasure channels \cite{erasure, eco}, which we refer to as pin maps. They possess vanishing classical and quantum capacities \cite{erasure}.

The primary contribution of this work is a novel protocol that overcomes this limitation by leveraging the quantum {\tt SWITCH} \cite{Chiribella2013Aug} in conjunction with sequential, unsharp measurements.
Over the last decade, the quantum {\tt SWITCH} has gathered a lot of attention in error mitigation providing advantages in various tasks ranging from quantum communication \cite{Chiribella2019, comswitch2, Chiribella2021}, quantum metrology \cite{sensing}, and quantum work extraction \cite{work1, work2, work3} to name a few. Many of these theoretical advantages have also been simulated in experiments as well \cite{exp1, exp2, exp3}. At the same time, unsharp measurements have also gained a lot of popularity for their role in tasks that involve probing the system sequentially. Some of the prominent examples of such tasks involve sequential violation of Bell inequalities \cite{ub1, ub2, ub3, ub4}, steering \cite{s1,s2}, witnessing entanglement \cite{ue1,ue2,ue3}, and even communication tasks like quantum teleportation \cite{Roy2021}.

 In our protocol, we demonstrate that by repeatedly performing weak (unsharp) measurements on the control qubit that preshares entanglement with the target system of the {\tt SWITCH}, one can certify a Bell violation at each stage, for multiple rounds while minimally disturbing the system, thereby preserving the quantum entanglement required for obtaining a violation in subsequent rounds. In particular, we identify that unlike projective measurements which yield maximal violation in a single deterministic round at the cost of all future nonlocality, uhsharp measurements introduce a trade-off of available Bell violation across access nodes. A weaker measurement at a given round yields a smaller violation but preserves more entanglement, enabling stronger violations in subsequent rounds. This trade-off can be utilized to control the available violation depending on the measurement sharpness.

 The quantum {\tt SWITCH} allows for the order of the noisy operations to be governed by a control qubit. 
  Indeed, a key component in the success of our protocol comes from the structure of the effective map after the {\tt SWITCH} action on two pin maps. While each pin map individually destroys all entanglement, the quantum {\tt SWITCH} combines them into a single operation whose resulting channel is no longer entanglement-breaking. The crucial effect of this {\tt SWITCH}ing action is the creation of a protected two-dimensional subspace. In our work, we consider the two erasure channels to pin every state to two orthogonal states that we label as $\ket0$ and $\ket1$ respectively. The quantum {\tt SWITCH} with these two pin maps possesses a decoherence-free subspace spanned by the vectors $\{ \ket{00}, \ket{11}\}$. This 
  right away implies that we can expect best success for initial states whose support in Bob and the Control's is restricted to the $\{ \ket{00}, \ket{11}\}$ subspace. To this end, we show that the SWITCH in the canonical form, i.e., with a coherent state in the control, is useless for the ``Random Access Bell Game." Our analysis reveals that a necessary condition for obtaining success in the Random Access Bell Game is entanglement between the quantum particles in possession of Alice and Bob, with the control of the quantum SWITCH.
  The requirement of entanglement along with the decoherence-free subspace spanned by $\{ \ket{00}, \ket{11}\}$ motivates us to choose the generalized GHZ state \cite{Mermin1990} as
 the prototypical initial state for which we demonstrate our results. Physically, unique multipartite correlations of a generalized GHZ state align perfectly with the protected subspace created by the {\tt SWITCH}, allowing it to survive the noisy evolution. 

In this paper, we derive the exact analytical form of the tripartite state after an arbitrary number of rounds, allowing us to precisely quantify the Bell violation as a function of the measurement sharpness. From this, we demonstrate that for an initial generalized GHZ state, a Bell violation can be maintained for an arbitrarily large number of rounds. We then show how to strategically exploit the identified trade-off by setting the measurement sharpness parameters in a geometric progression. This strategy allows a Bell violation arbitrarily close to the Tsirelson bound \cite{Cirelson1980} to be achieved at any desired target round, while still guaranteeing a violation in all preceding rounds. We complement these analytical findings with a numerical analysis that quantifies the persistence of nonlocality, calculating the maximum number of rounds for which a specific, predefined threshold of Bell violation can be guaranteed. We finally contrast the success of the generalized GHZ state to the case where we consider the initial state between Alice, Bob and the control to be a W-state \cite{durvidalcirac}.
We prove that the protocol fails completely when the initial resource is a W state. The W state's distinct correlation structure is not preserved by the {\tt SWITCH} action, causing the state to become fully separable after the first round and thus providing a clear operational distinction between the two SLOCC (Stochastic Local Operations and Classical Communication) inequivalent classes of three qubit pure states \cite{durvidalcirac}.

The paper is organized as follows. After this introduction in Sec.~\ref{sec:intro}, we formally define the Random Access Bell Game in Sec.~\ref{sec:game}. In this section, we also analyze preliminary strategies, demonstrating the limitations of canonical approaches that do not involve pre-shared entanglement with the control system. Our main protocol, which leverages the quantum {\tt SWITCH} in conjunction with sequential unsharp measurements and a GHZ-type initial state, is presented in Sec.~\ref{sec:protocol_analysis}. We then generalize our results to an arbitrary number of rounds in Sec.~\ref{sec:k_round}, where we provide the central proofs that Bell violation can be maintained indefinitely and that near-maximal violation is achievable at any chosen round. In Sec.~\ref{sec:w_state}, we highlight the necessity of GHZ-type entanglement by proving the protocol's failure for W states. Finally, we summarize our findings and discuss their broader implications in Sec.~\ref{sec:conclusion}. Detailed proofs and calculations are deferred to the Appendices along with a brief section on preliminaries.

\section{Random Access Bell Game}
\label{sec:game}
\begin{figure}[t]
    \centering
    \includegraphics[width=\linewidth]{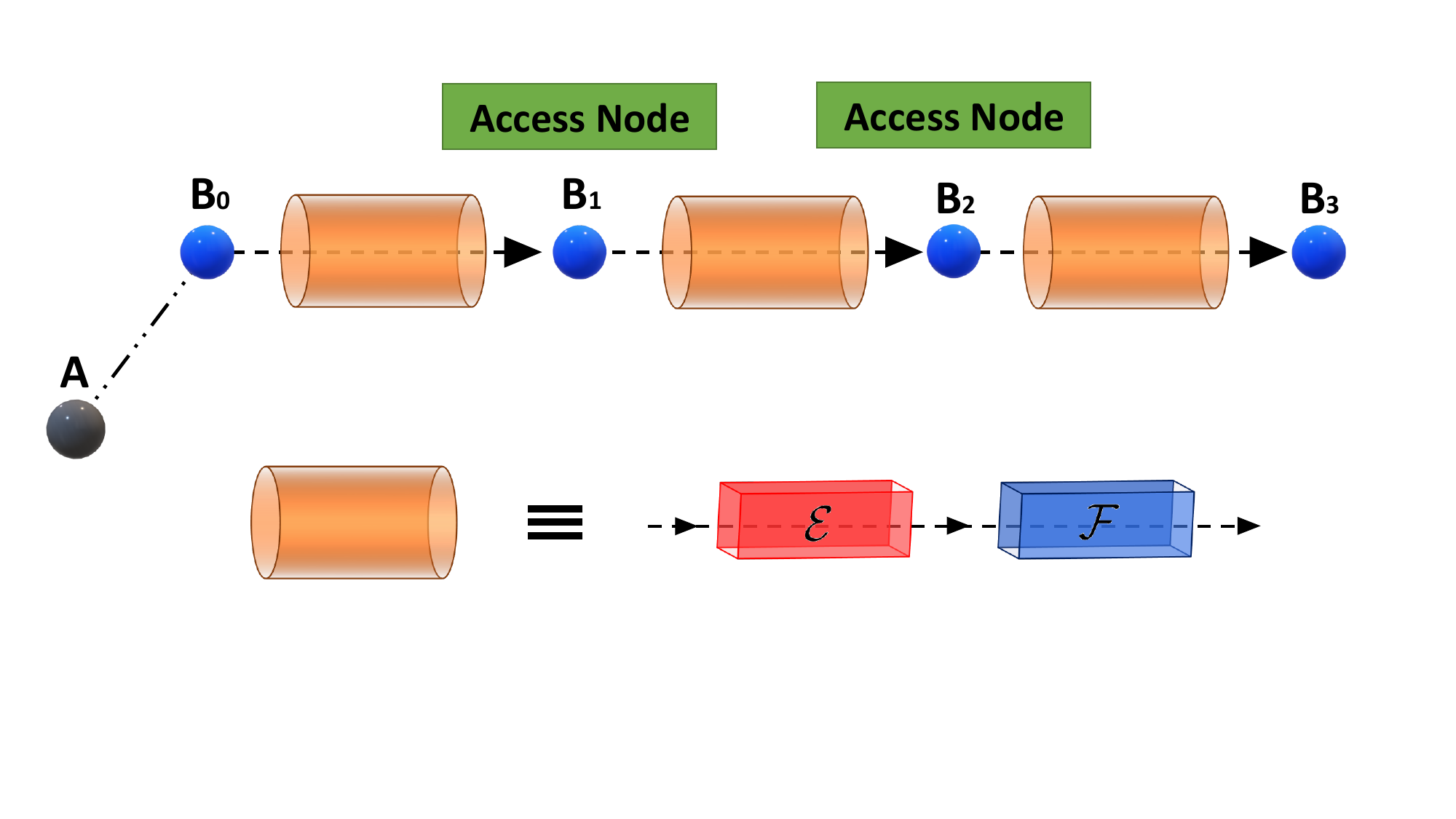}
    \caption{\textbf{Schematic of the random access Bell game.} The subsystem at B repeatedly passes through identical noise blocks. The particle can be picked randomly from any access nodes and used to assess the Bell violation.}
    \label{fig:enter-label}
\end{figure}
The setting of the game involves a pre-shared entangled state $\rho_{AB_0}$ between Alice ($A$) and Bob ($B$). The subscript $0$ denotes the initial setting. Now the particle with Bob is transmitted through a series of identical noisy channels, say, $\Lambda$. The output state after passing through $N$ iterations of the channel is denoted by $\rho_{AB_N} := (\mathbb I \otimes \Lambda)^{\otimes N}(\rho_{AB_0})$. The subsystem $B_k$  can be accessed for measurement after any round $k$ of channel actions. A third party, which we refer to as the referee $R$, \emph{randomly} dictates after how many channel actions $k$ the state can be extracted. We name the locations after each channel action from which the particle can be extracted as an \emph{access node}. Now the central question is whether $\rho_{AB_k}$ violates a Bell inequality. Or more generally, what is the maximum value $k_{\max}$ up to which violation is possible when the particle is randomly extracted from any access node $k\leq k_{\max}.$

If $\Lambda$ is a perfect (noiseless) channel, then $k_{\max}$ can be arbitrarily large, provided the initial state $\rho_{AB_0}$ violates, in particular, the CHSH inequality \cite{Clauser1972}. This inequality tests for nonlocality by constraining the correlations between local measurements; a violation occurs when the measured Bell parameter $|\mathcal{B}|$ exceeds the classical bound of 2, which is impossible under local realism but permitted by quantum mechanics. For a detailed description of the formalism, please see Appendix \ref{sec:preliminaries}.

The situation becomes interesting when $\Lambda$ is noisy.
For example, when $\Lambda$ is an entanglement-breaking channel, $k_{\max} = 0$ since $\mathbb I \otimes\Lambda (\rho_{AB_0}) = \rho_{AB_1}$ is a separable state and hence does not violate a Bell inequality.
The particular noisy channel we consider is one where $\Lambda$ is composed of entanglement-breaking pin maps, $\map E$ and $\map F$ whose actions are given as follows
\begin{eqnarray}
    \map E(\rho) = \ketbra{0}{0} ~\forall \rho, ~~~\map F(\rho) = \ketbra{1}{1} ~\forall \rho.
\end{eqnarray}
The noisy channel is described by a probabilistic mixture (a causally separable configuration) of the channels $\map E$ and $\map F$
\begin{eqnarray}
    \Lambda_p = p ~\map E \circ \map F + (1-p) ~\map F \circ \map E = p ~\map E + (1-p) ~\map F.
\end{eqnarray}
Even if there were a control system $C$ that possessed the power to tune the probability $p$, the output state of the channel would still be separable, 
\begin{eqnarray}
    \mathbb{I}\otimes\Lambda_p(\rho_{AB_0}) = \rho_A \otimes \big(p \ketbra{0}{0}+(1-p)\ketbra{1}{1}\big),
\end{eqnarray}
where $\rho_A = \Tr_B(\rho_{AB_0})$ is the marginal state of $A$. Now the key point of investigation is to consider the case where $C$ can coherently control the order of the channels $\map E$ and $\map F$. This setup, which we refer to as the {\tt {\tt SWITCH}} action (see Fig. \ref{fig:enter-label2}), also turns out to be of limited utility in its usual configuration. 
\begin{figure}[ht]
    \centering
    \includegraphics[width=\linewidth]{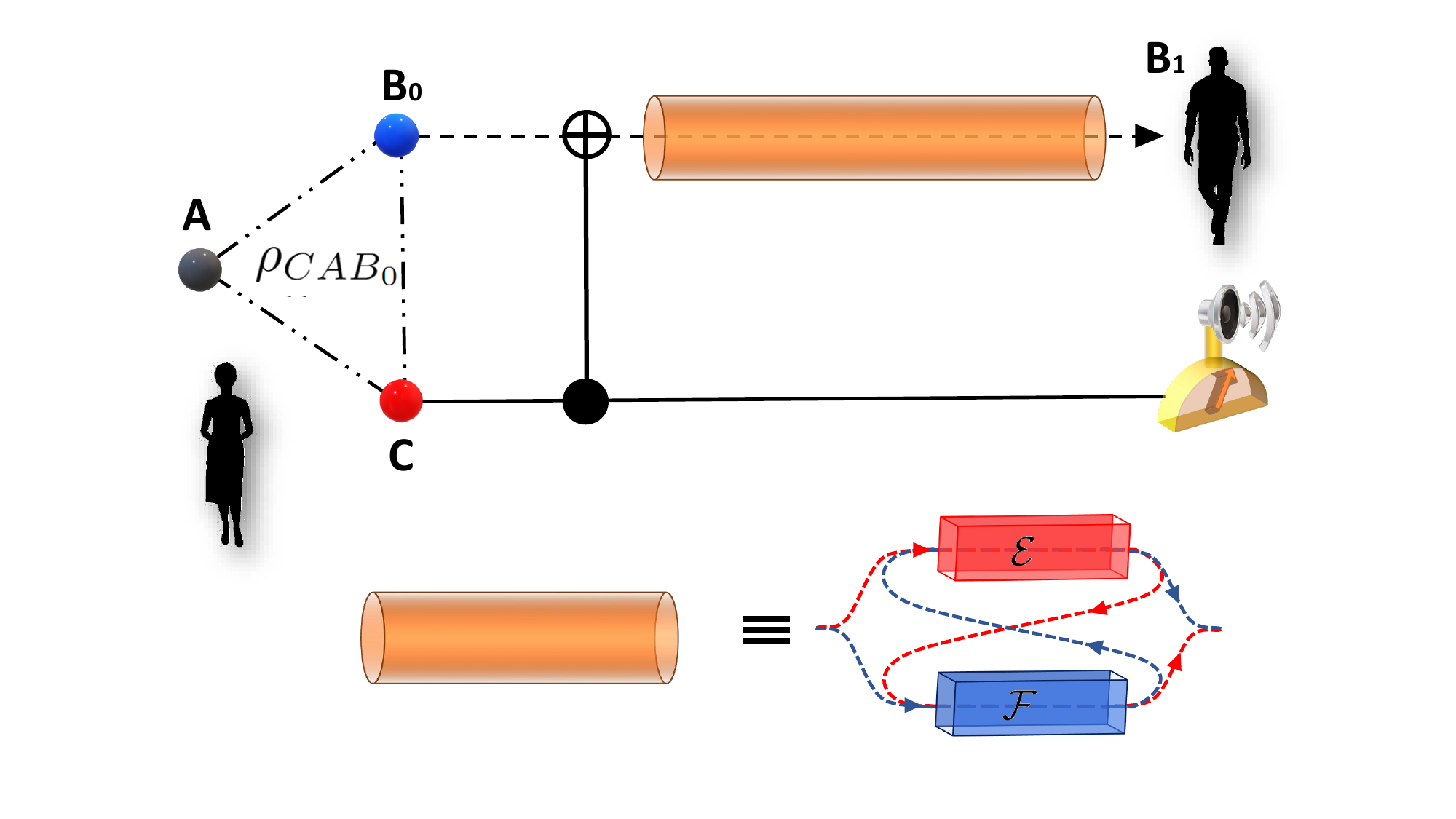}
    \caption{\textbf{Random Access Bell Game using the quantum {\tt SWITCH}}. The initial state $\rho_{CAB_0}$ is passed through the coherently controlled noise, e.g., the {\tt SWITCH} action, with C acting as the control. A measurement is then performed at C and the outcomes are communicated to Alice and Bob.}
    \label{fig:enter-label2}
\end{figure}

%\subsection{{\tt {\tt SWITCH}} Action}
The {\tt {\tt SWITCH}} action channel which results from the quantum {\tt {\tt SWITCH}} supermap acting on channels $\mathcal{E}$ and $\mathcal{F}$ is elaborated below. 
The Kraus operators of the noisy entanglement-breaking channels $\mathcal{E}$ and $\mathcal{F}$ are $\mathcal{E}_1 = \ket{0}\bra{0}$, $\mathcal{E}_2 = \ket{0}\bra{1}$ and $\mathcal{F}_1 = \ket{1}\bra{0}$, $\mathcal{F}_2 = \ket{1}\bra{1}$. The quantum {\tt {\tt SWITCH}} supermap takes these as inputs and outputs the channel $\text{{\tt {\tt SWITCH}}}(\mathcal{E}, \mathcal{F})$ with Kraus operators $\mathcal{S}_{ij} = \ket{0}\bra{0} \otimes \mathcal{E}_i \mathcal{F}_j + \ket{1}\bra{1} \otimes \mathcal{F}_j \mathcal{E}_i$, while Alice's subsystem $A$ remains unaffected. If we denote the combined channel over the three-party system as $\mathcal{K}$, we get the following Kraus operators:

\begin{eqnarray}
K_{0} = \mathcal{S}_{11} \otimes \mathbb{I}_A &=& \ket{1}\bra{1}_C \otimes \mathbb{I}_A\otimes \ket{1}\bra{0}_B, \notag\\ 
K_{1} = \mathcal{S}_{12} \otimes \mathbb{I}_A &=& 0, \notag\\
K_{2} = \mathcal{S}_{21} \otimes \mathbb{I}_A &=& \ket{1}\bra{1}_C \otimes \mathbb{I}_A \otimes \ket{1}\bra{1}_B \notag\\ &&+ \ket{0}\bra{0}_C \otimes \mathbb{I}_A \otimes \ket{0}\bra{0}_B, \notag\\ 
K_{3} = \mathcal{S}_{22} \otimes \mathbb{I}_A &=& \ket{0}\bra{0}_C \otimes \mathbb{I}_A \otimes \ket{0}\bra{1}_B.
\label{eqn:Switch_Kraus}
\end{eqnarray}
Any arbitrary three-party state $\rho_{CAB}$ under the {\tt {\tt SWITCH}} action gets updated as $\mathcal{K}(\rho_{CAB}) = \sum_{i=0}^3 K_i \rho_{CAB}K_i^\dagger$.

\begin{lemma}
    The quantum {\tt {\tt SWITCH}} with a canonical setup of employing a coherent state at the control is useless for the random access Bell game.
    \label{lemma:1}
\end{lemma}
\begin{proof}
We start with the most general configuration of the initial state, with a coherent control 
\begin{eqnarray}
    \rho_{CAB_0} = \ket{+}_C\otimes \ket{\phi^+_{\alpha}}_{AB_0},
\end{eqnarray}
where $\ket{\phi^+_{\alpha}}= \sqrt{\alpha}\ket{00}+ \sqrt{1-\alpha}\ket{11}.$
The output state after the {\tt {\tt SWITCH}} action is given by 
    \begin{eqnarray}
        \rho_{CAB_1} = \frac12 \ketbra{\text{GHZ}_\alpha}{\text{GHZ}_\alpha} &+& \frac{1-\alpha}{2}\ketbra{010}{010} \nonumber \\ &+& \frac{\alpha}{2} \ketbra{101}{101}), \nonumber
    \end{eqnarray}
    where $\ket{\text{GHZ}_\alpha} = \sqrt{\alpha}\ket{000}+ \sqrt{1-\alpha}\ket{111}).$
    After the control $C$ is measured in the $\ket \pm$ basis, the un-normalized output states are computed to be
 \begin{eqnarray}
        \widetilde{\rho}^{\pm}_{AB_1} = \frac14 \big(\ket{\phi^+_{\alpha}}\bra{\phi^+_{\alpha}} +(1-\alpha)\ketbra{10}{10}+ \alpha \ketbra{01}{01}\big).
    \end{eqnarray}
   Now, we show that  $\widetilde{\rho}^{\pm}_{AB_1}$ is separable. Since $\widetilde{\rho}^{\pm}_{AB_1}$ we use the Postive Partial Transpose (PPT) criterion \cite{ppt} which is necessary and sufficient for $2\otimes 2$ and $2\otimes 3$ systems. Now, it is straightforward to show
   \begin{eqnarray}
(\widetilde{\rho}^{\pm}_{AB_1})^{T_A} \geq 0 ~~\forall \alpha,
   \end{eqnarray}
   since the eigenvalues of $ (\widetilde{\rho}^{\pm}_{AB_1} )^{T_A}$ are $\{1,0,\alpha,1-\alpha\}$.
This, in turn, guarantees the separability of $\widetilde{\rho}^{\pm}_{AB_1}$. Hence the proof.
\iffalse
    We provide a sketch of the proof here with an example. Suppose the initial state is $\ket{+}_C\otimes \ket{\phi^+}_{AB_0}$ ~~$\ket{+}_C\otimes \ket{\phi^+_{\alpha}}_{AB_0}$. The output state after the {\tt {\tt SWITCH}} action is given by 
    \begin{eqnarray}
        \rho_{CAB_1} = \frac12 \ketbra{\text{GHZ}}{\text{GHZ}} + \frac{1}{4}(\ketbra{010}{010} + \ketbra{101}{101}), \nonumber
    \end{eqnarray}
    where $\ket{\text{GHZ}} = \frac{1}{\sqrt{2}}(\ket{000}+\ket{111}).$
    After the control $C$ is measured in the $\ket \pm$ basis, the un-normalized output states are computed to be
    \begin{eqnarray}
        \widetilde{\rho}^{\pm}_{AB_1} = \frac18 (\ketbra{\Phi^\pm}{\Phi^\pm}+\ketbra{\Psi^+}{\Psi^+}+\ketbra{\Psi^-}{\Psi^-}).
    \end{eqnarray}
    Both $ \widetilde{\rho}^{\pm}_{AB_1}$ are separable Bell diagonal states, and therefore cannot violate a Bell inequality. The general proof (in Appendix X) \ga{add the general proof} involves an arbitrary state $\rho_{AB_0}$ instead of $\ket{\Phi^+}_{AB_0}$ as the initial state. 
    \fi
\end{proof}

\begin{lemma}
    A GHZ state shared between $CAB_0$  can enable maximal violation of the CHSH inequality when the particle is deterministically/conclusively picked up from any access node $k$ by $B_k$.
\label{lemma:lemma2}
\end{lemma}
\begin{proof}
    It is evident from Eq. \eqref{eqn:Switch_Kraus} that the {\tt {\tt SWITCH}} action preserves any state in the subspace of $\{\ket{000}, \ket{111}\}$. %because while the $\mathcal{K}_1$ and $\mathcal{K}_3$ operators annihilate such states, the $\mathcal{K}_2$ operator leaves them unchanged. Thus, 
    Therefore, initially starting with a GHZ state, $\rho_{CAB_0} =\ket{\text{GHZ}}\bra{\text{GHZ}}$ from Eq. \eqref{eqn:Switch_Kraus}, after any arbitrary but pre-declared round $k$, we get
\begin{eqnarray}
\rho_{CAB_k} = \mathcal{K}^{\otimes k}(\rho_{CAB_0}) =\ket{\text{GHZ}}\bra{\text{GHZ}}.
\label{eqn:rhocabk}
\end{eqnarray}
After $k$ iterations, the control $C$ measures its part of the system in the $\ket{\pm}$ basis. Then, up to local unitaries, the state shared between $A$ and $B_k$ is $\ket{\Phi^+}$ that violates the CHSH inequality maximally.
\end{proof}

Things become different when the particle is picked up from a random node. Then, the quantity of interest becomes $k_{\max} =:$ the maximal round up to which, if the particle is extracted, demonstrates violation of the CHSH inequality for all rounds till $k_{\max}$.

\begin{observation}
    The random access Bell game \textbf{fails} for any strategy involving projective measurements at $C$ at any round $m$ other than the randomly selected round.
\end{observation}
\begin{proof}
    Projective measurements at the control $C$ at any round $m$ result in a product state across the $C:AB_m$ cut. Following Lemma \ref{lemma:1}, we conclude the absence of CHSH violation for all subsequent rounds.
\end{proof}

The objective of the paper is to construct a protocol that enables success in the random access Bell game, by demonstrating a finite number of rounds $(k >1)$ for which CHSH violation persists if the particle is picked from any access node $m \leq k$.
%Although the above scheme performs well when the round $k$ at which the particle is picked up to check Bell violations, since the target and control become unentangled, specifically a product state, from Lemma. 1 it follows that if the particle was picked from any other access node there would be no violation. Now, a crucial question surfaces. What if the referee prescribes to pick up the particle from an access node $k >1$. In other words, 
Specifically, we want to investigate  \emph{what is the furthest access node $k_{\max}$ such that if the particle was picked from any access node $k \leq k_{\max}$, there is a guaranteed violation of the CHSH inequality between $A$ and $B_k$.}

\section{Advantages in the random access Bell Game using the quantum {\tt {\tt SWITCH}}}
\label{sec:protocol_analysis}
%\noindent \textcolor{red}{change all $max$ to $\max$.}\\
%\textcolor{red}{change $\ket{\Phi^+_\alpha} \to \ket{\Phi^+_\alpha}$, and $a \to \alpha$} \ga{done!}\\

We now present a protocol that uses the quantum {\tt {\tt SWITCH}} and sequential unsharp measurements of the control to enable success in the random access Bell game. Mathematically, success implies that we can achieve $k_{\max} > 1$.
\begin{protocol}\emph{(Sequential strategy.)}
  Alice and Bob initially share a tripartite entangled state, $\rho_{CAB_0}$. The protocol, which guarantees CHSH violation for all rounds $k \leq k_{\max},$ proceeds in the following steps. 
    \begin{enumerate}
    \item The combined three qubit state undergoes the coherently controlled noise, i.e., the {\tt {\tt SWITCH}} action, with subsystem C set as the control. 
        \item  After the {\tt {\tt SWITCH}} action, we obtain the state $\rho_{CAB_1}$. If $k_{\max}>1$, Alice performs   an unsharp measurement on the control system $C$ described by the POVMs $\{E_{\pm}^{\lambda_k}\}$ to yield one of the post-measurement states $\rho_{CAB_1}^{E_\pm}$, where
\begin{eqnarray}
            E_{\pm}^{\lambda_k} &=& \lambda_k \ketbra{\pm}{\pm} + \frac{1-\lambda_k}{2} \mathbb I \nonumber \\
            &=& \frac{1+\lambda_k}{2} \ketbra{\pm}{\pm} + \frac{1-\lambda_k}{2}\ketbra{\mp}{\mp},
            \label{eqn:unsharp_measure}
        \end{eqnarray}
      and where $\lambda_k$ is the sharpness parameter in round $k$. The sharpness parameter at any given round $k$ is chosen to ensure Bell violation, i.e.,  $|\map B(\rho_{AB_k})| \geq 2+ \delta$ for some positive number $\delta \in (0,2\sqrt{2}-2]$.

 If $k_{\max} =1$, $C$ is measured projectively in the $\{\ket{+},\ket{-}\}$ basis to maximize $\map B(\rho_{AB_{k_{\max}}})$. 

 \item Steps 1. and 2. are repeated. After $k$ rounds, states $\rho_{CAB_k}$  (post-{\tt {\tt SWITCH}} action) and $\rho_{CAB_k}^{E_\pm}$ (post-measurement)  are obtained. This continues until the final round $k_{\max}th$, at which Alice performs a projective measurement on $C$ and communicates the measurement outcome to Bob. 
    \end{enumerate}
\end{protocol}

\begin{comment}
$\rho_{AB_0}= \ketbra{\phi^+}{\phi^+}$, where $\ket{\phi^+} = \frac{1}{\sqrt{2}}(\ket{00}+\ket{11})$.  The protocol proceeds in the following steps. 
    \begin{enumerate}
        \item Alice takes a ancillary particle $C$ in the $\ket{+}$ state and performs a {\tt CNOT} operation using the qubit in her possession. $C$ is used as the control of the quantum {\tt {\tt SWITCH}}. In effect, now $C$, $A$ and $B_0$ share a GHZ-state $\rho_{CAB_0} = \ketbra{\text{GHZ}}{\text{GHZ}}, ~\emph{where} ~\ket{\text{GHZ}}=\frac{1}{\sqrt{2}}(\ket{000}+\ket{111})$.

        \item The system B in possesion of Bob, undergoes the coherently controlled noise, e.g. the {\tt {\tt SWITCH}} action with control state C
     
        \item After the {\tt {\tt SWITCH}} action, we end up with the state $\rho_{CAB_1}$. Now two cases arise. If $k_{max} = 1$, Alice performs a projective measurement $\{\ket+,\ket-\}$ on her $C$ subsystem and communicates the measurement outcome to $B_1$. If $k_{max}\neq1$, then she performs the unsharp measurement on $C$ as described in Eq.~\eqref{eqn:unsharp_measure} leading to either of the post measurement states $\rho_{CAB_k}^{E\pm}$.

        \item Steps 2. and 3. are repeated. After $k$ rounds of the protocol states $\rho_{CAB_k}$ and $\rho_{CAB_k}^{E_\pm}$  are obtained. This is done until the $k_{max}th$ step, at which Alice performs a projective measurement on $C$ and communicates the measurement outcome to $B_{k_{max}}$. 
    \end{enumerate}

\end{comment}
In the subsequent section, we implement the protocol and evaluate the success metrics in the ``Random Access Bell Game" when $\rho_{CAB_0}$ is a generalized GHZ state.

\subsection{Persistent violation for arbitrarily many rounds}

Let us establish an instance of the above protocol where we start with a generalized GHZ state $\rho_{CAB_0} = \ketbra{\text{GHZ}_\alpha}{\text{GHZ}_\alpha}$, where $\ket{\text{GHZ}_\alpha}=\sqrt{\alpha}\ket{000}+\sqrt{1-\alpha}\ket{111}$.

\begin{comment}
    \begin{enumerate}
        \item Alice initializes the ancillary particle $C$ in the state $\ket{0}$ and performs a {\tt CNOT} operation on $C$ using the qubit in her possession as the control. In effect, $C$, $A$, and $B_0$ now share a generalized GHZ state, $\rho_{CAB_0} = \ketbra{\text{GHZ}_\alpha}{\text{GHZ}_\alpha}$, where $\ket{\text{GHZ}_\alpha}=\sqrt{\alpha}\ket{000}+\sqrt{1-\alpha}\ket{111}$.

         \item The system $B$, in Bob's possession, undergoes the coherently controlled noise, i.e., the {\tt {\tt SWITCH}} action with control system $C$.
     
        \item After the {\tt {\tt SWITCH}} action, we obtain the state $\rho_{CAB_1}$. Two cases arise. If $k_{\max} = 1$, Alice performs a projective measurement on her $C$ subsystem in the $\{\ket{+},\ket{-}\}$ basis and communicates the outcome to Bob. If $k_{\max} \neq 1$, she performs an unsharp measurement on $C$ as described in Eq.~\eqref{eqn:unsharp_measure}, leading to one of the post-measurement states $\rho_{CAB_1}^{E_\pm}$.

        \item Steps 2. and 3. are repeated. After $k$ rounds, states $\rho_{CAB_k}$  (post-{\tt {\tt SWITCH}} action) and $\rho_{CAB_k}^{E_\pm}$ (post-measurement)  are obtained. This continues until the final round $k_{\max}th$, at which Alice performs a projective measurement on $C$ and communicates the measurement outcome to Bob. 
    \end{enumerate}

With our protocol established, let us evaluate the outcomes of the first few rounds.
\end{comment}

\subsubsection{Round 1}

\textit{{\tt SWITCH} Action.} Since the generalized GHZ state lies in the protected subspace of $\{\ket{000}, \ket{111}\}$ from Eq. \eqref{eqn:Switch_Kraus}, {\tt SWITCH} action on the state $\rho_{CAB_0} = \ketbra{GHZ_\alpha}{GHZ_\alpha}$ leaves the state unchanged
\begin{equation}
\rho_{CAB_1} = \mathcal{K}(\rho_{CAB_0}) =\ketbra{GHZ_\alpha}{GHZ_\alpha}. 
\label{eqn:rhocab1}
\end{equation}

\textit{Unsharp Measurement.} To evaluate the post-measurement states of the unsharp measurement, the POVM elements are expressed in terms of measurement operators $E^{\lambda_k}_\pm = \sqrt{E^{\lambda_k}_\pm }^\dagger \sqrt{E^{\lambda_k}_\pm }$, with
\begin{eqnarray}
\sqrt{E_\pm^{\lambda_k}} = \sqrt{\frac{1\pm \lambda_k}{2}}\ket{+}\bra{+} + \sqrt{\frac{1\mp \lambda_k}{2}}\ket{-}\bra{-}.
\end{eqnarray}
With the GHZ state obtained in the first round, $\rho_{CAB_1} = \ket{\text{GHZ}_\alpha}\bra{\text{GHZ}_\alpha}$, a post-measurement state $\rho_{CAB_1}^{E_\pm}$ for the combined system, depending on the measurement outcome, takes the form
\begin{eqnarray}
&&\rho_{CAB_1}^{E_\pm} = \frac{\left(\sqrt{E^{\lambda_1}_\pm }\otimes \mathbb{I}_{AB}\right)\rho_{CAB_1}\left(\sqrt{E^{\lambda_1}_\pm }\otimes \mathbb{I}_{AB}\right)^\dagger}{\Tr[\rho_{CAB_1}(E^{\lambda_1}_\pm \otimes \mathbb{I}_{AB})]} \notag\\
&&=\frac{1\pm \lambda_1}{2}\ket{+\Phi^+_\alpha}\bra{+\Phi^+_\alpha} + \frac{1\mp \lambda_1}{2}\ket{-\Phi^-_\alpha}\bra{-\Phi^-_\alpha} \notag\\
&&+ \frac{\sqrt{1-\lambda_1^2}}{2}\left(\ket{+\Phi^+_\alpha}\bra{-\Phi^-_\alpha}+\ket{-\Phi^-_\alpha}\bra{+\Phi^+_\alpha}\right),
\label{eqn:rhoEcab1}
\end{eqnarray}

where the generalized Bell states are defined as
\begin{eqnarray}
\ket{\Phi^\pm_\alpha} &=& \sqrt{\alpha}\ket{00}\pm\sqrt{1-\alpha}\ket{11},\notag\\
\ket{\Psi^\pm_\alpha} &=& \sqrt{\alpha}\ket{01}\pm\sqrt{1-\alpha}\ket{10},
\end{eqnarray}
and the sign $\pm$ depends on the measurement outcome. We also use the standard shorthand $\ket{\Phi^+} :=\ket{\Phi^+_{\alpha= \frac12}}$. 
 The obtained state is a probabilistic mixture of two generalized Bell states with maximum mixing at $\lambda_1 = 0$, which decreases as $\lambda_1$ grows.

\textit{Bell Violation.} We can now evaluate the Bell violation on the A:B cut. Tracing out subsystem $C$ from the post-measurement state $\rho_{CAB_1}^{E_\pm}$ yields
\begin{eqnarray}
\rho_{AB_1}^{E_\pm}  = \frac{1\pm\lambda_1}{2}\ket{\Phi^+_\alpha}\bra{{\Phi^+_\alpha}} +\frac{1\mp\lambda_1}{2}\ket{\Phi^-_\alpha}\bra{{\Phi^-_\alpha}}.
\label{eqn:rhoEcab1} \notag
\end{eqnarray}
Horodecki's criterion is particularly straightforward for generalized Bell states because of the simple form of the $T_\rho$ matrix for these states and the fact that its components are linear in $\rho$. From the $T_\rho$ matrices for such states, 
\begin{eqnarray}
T_{\Phi^+_\alpha} &=& \mathrm{\text{Diag}}\left(2\sqrt{\alpha(1-\alpha)},-2\sqrt{\alpha(1-\alpha)},1\right), \notag\\ T_{\Phi^-_\alpha} &=& \mathrm{\text{Diag}}\left(-2\sqrt{\alpha(1-\alpha)},2\sqrt{\alpha(1-\alpha)},1\right),\notag\\ T_{\Psi^+_\alpha} &=& \mathrm{\text{Diag}}\left(2\sqrt{\alpha(1-\alpha)},2\sqrt{\alpha(1-\alpha)},-1\right), \notag\\ T_{\Psi^-_\alpha} &=& \mathrm{\text{Diag}}\left(-2\sqrt{\alpha(1-\alpha)},-2\sqrt{\alpha(1-\alpha)},-1\right), \notag
\end{eqnarray}
we get
$$T_{\rho_{AB_1}^{E_\pm}} = \mathrm{\text{Diag}}\left(\pm2\sqrt{\alpha(1-\alpha)}\lambda_1, \mp2\sqrt{\alpha(1-\alpha)}\lambda_1, 1\right),$$ 
thus yielding an $M(\rho)$ value of $M(\rho) = 1+4\alpha(1-\alpha)\lambda_1^2$ and, therefore, a Bell violation in the first round of $|\mathcal{B}_1|= 2\sqrt{1+ 4\alpha(1-\alpha)\lambda_1^2}$, independent of which post-measurement outcome is obtained.  As expected, the maximum violation occurs with GHZ and Bell states ($\alpha = 1/2$) and decays as $\alpha$ approaches 0 or 1. We conclude that a Bell violation will occur in the first round for all $\alpha \in (0,1)$ and unsharp measurements with $\lambda_1 >0$.\\

\subsubsection{Round 2}

\textit{{\tt SWITCH} Action.} Starting with the post-measurement states $\rho_{CAB_1}^{E_\pm}$ from round one, the second round of the protocol begins with an application of the {\tt {\tt SWITCH}} action, which results in the state
\begin{eqnarray}
\rho_{CAB_2} &=& \sum_{i = 1}^3 \mathcal{K}_i\rho_{CAB_1}^{E_\pm} \mathcal{K}_i^\dagger\notag\\
&&  \hspace{-1cm}= \frac{1+\sqrt{1-\lambda_1^2}}{2}\ket{\text{GHZ}_\alpha}\bra{\text{GHZ}_\alpha} \notag\\ 
&& \hspace{-1cm}+ \frac{1-\sqrt{1-\lambda_1^2}}{4}\bigg(2(1-\alpha)\ket{010}\bra{010} + 2\alpha\ket{101}\bra{101}\bigg). \notag
\label{eqn:rhocab2}
\end{eqnarray}
Interestingly, the resulting state is independent of the first round's measurement outcome, as the {\tt SWITCH} action maps both possible post-measurement states $\rho_{CAB_1}^{E+}$ and $\rho_{CAB_1}^{E-}$ to the same output state. As such, we do not have to keep track of the particular post-measurement outcomes. The result of performing unsharp measurement in the first round is now apparent in the form of states $\ketbra{010}{010}$ and $\ketbra{101}{101}.$ Had we not performed the measurement (corresponding to the trivial case $\lambda_1  =0$), the state would have remained a generalized GHZ state. Tracing out system C from this would result in a separable state for A and B, yielding no Bell violation.\\ 

\textit{Unsharp Measurement.} After the {\tt {\tt SWITCH}} action, an unsharp measurement is performed on the control qubits C. Utilizing POVMs $E_\pm^{\lambda_2}$ and measurement operators $\sqrt{E_\pm^{\lambda_2}}$ leads to a post-measurement state
\begin{align}
\rho_{CAB_2}^{E_\pm}
&= \Big[\tfrac{1+\sqrt{1-\lambda_1^2}}{2}\Big]
   \Bigg(
     \tfrac{1\pm\lambda_2}{2}
       \ket{+\Phi^+_\alpha}\bra{+\Phi^+_\alpha} \notag\\
&\hspace{3cm}
   + \tfrac{1\mp\lambda_2}{2}
       \ket{-\Phi^-_\alpha}\bra{-\Phi^-_\alpha} \notag\\
&\hspace{3cm}
   + \tfrac{\sqrt{1-\lambda_2^2}}{2}
       \ket{+\Phi^+_\alpha}\bra{-\Phi^-_\alpha} \notag\\
&\hspace{3cm}
   + \tfrac{\sqrt{1-\lambda_2^2}}{2}
       \ket{-\Phi^-_\alpha}\bra{+\Phi^+_\alpha}
   \Bigg) \notag\\[4pt]
&\quad + \Big[\tfrac{1-\sqrt{1-\lambda_1^2}}{4}\Big]
   \Bigg(
     \tfrac{1\pm\lambda_2}{2}
       \ket{+\Psi^+_\alpha}\bra{+\Psi^+_\alpha} \notag\\
&\hspace{3cm}
   + \tfrac{1\pm\lambda_2}{2}
       \ket{+\Psi^-_\alpha}\bra{+\Psi^-_\alpha} \notag\\
&\hspace{3cm}
   + \tfrac{1\mp\lambda_2}{2}
       \ket{-\Psi^+_\alpha}\bra{-\Psi^+_\alpha} \notag\\
&\hspace{3cm}
   + \tfrac{1\mp\lambda_2}{2}
       \ket{-\Psi^-_\alpha}\bra{-\Psi^-_\alpha}
   \Bigg) \notag\\[4pt]
&\quad + \Big[\tfrac{1-\sqrt{1-\lambda_1^2}}{4}\tfrac{\sqrt{1-\lambda_2^2}}{2}\Big]
   \Bigg(
     (1-\alpha)\ket{010}\bra{010} \notag\\
&\hspace{1cm}
   + \alpha\ket{101}\bra{101}
   - (1-\alpha)\ket{110}\bra{110} \notag\\
&\hspace{2cm}
   - \alpha\ket{001}\bra{001}
   \Bigg),
\label{eqn:rhoEcab2}
\end{align} where the sign $\pm$ depends on the particular post-measurement outcome, and $\ket{\Phi^\pm_\alpha}, \ket{\Psi^\pm_\alpha}$ are the appropriate generalized Bell states.

\textit{ Bell Violation.} Tracing out the system C to access the Bell violation on the A:B cut, we get 

\begin{eqnarray}
\rho_{AB_2}^{E_\pm} &=&\Bigg[\frac{1+\sqrt{1-\lambda_1^2}}{2}\Bigg]\Bigg(\frac{1\pm \lambda_2}{2}\ket{\Phi^+_\alpha}\bra{\Phi^+_\alpha} \notag\\ && \hspace{3cm}+ \frac{1\mp \lambda_2}{2}\ket{\Phi^-_\alpha}\bra{\Phi^-_\alpha} \Bigg) \notag\\
&+&\Bigg[\frac{1-\sqrt{1-\lambda_1^2}}{4}\Bigg]\Bigg(\ket{\Psi^+_\alpha}\bra{\Psi^+_\alpha} + \ket{\Psi^-_\alpha}\bra{\Psi^-_\alpha}\Bigg). \notag\\
\end{eqnarray}
Since the state is in the Bell basis, the $T_{\rho_{AB_2}^{E_\pm}}$ matrix is easily calculated: 
\begin{eqnarray}
T_{\rho_{AB_2}^{E_\pm}} &&= \mathrm{\text{Diag}}\left(\pm\sqrt{\alpha(1-\alpha)}\left[1+\sqrt{1-\lambda_1^2}\right]\lambda_2, \right. \notag\\ 
&& \left. \mp\sqrt{\alpha(1-\alpha)}\left[1+\sqrt{1-\lambda_1^2}\right]\lambda_2, \sqrt{1-\lambda_1^2}\right),
\end{eqnarray}
that leads to a Bell violation of 
\begin{eqnarray}
|\mathcal{B}_2| = \begin{cases} 2\sqrt{2\alpha(1-\alpha)}\left[1+\sqrt{1-\lambda_1^2}\right]\lambda_2,\\  \hspace{0.2cm} \text{if }  \sqrt{\alpha(1-\alpha)}\left[1+\sqrt{1-\lambda_1^2}\right]\lambda_2 \geq \sqrt{1-\lambda_1^2}, \\\\
2\sqrt{\alpha(1-\alpha)\left[1+\sqrt{1-\lambda_1^2}\right]^2\lambda_2^2 +1-\lambda_1^2}, \\ \hspace{0.2cm} \text{if }  \sqrt{\alpha(1-\alpha)}\left[1+\sqrt{1-\lambda_1^2}\right]\lambda_2 \leq \sqrt{1-\lambda_1^2}.
\end{cases}
\label{eqn:Bellviol_round2}
\end{eqnarray}

As in the first round, the violation is independent of which of the post-measurement outcomes was obtained. The violation decreases as $\alpha$ approaches 0 or 1, with the maximum obtained for the GHZ state ($\alpha = 1/2$). Note that setting $\alpha = 1/2$, $\lambda_2 = 1$, and $\lambda_1 = 0$ yields a maximal Bell violation of $2\sqrt{2}$ in the second round, but this would give no violation in the first round. There is a certain trade-off in the allowed violation since $|\mathcal{B}_1|$ increases with $\lambda_1$, while $|\mathcal{B}_2|$ decreases with $\lambda_1$ and increases with $\lambda_2$. Indeed, if we set $\lambda_1 = 1$ (a projective measurement), a maximum violation occurs in round 1, but no entanglement is preserved for subsequent rounds. As we decrease $\lambda_1$, the measurement becomes unsharp, leading to less violation in the first round and more violation in the second round. At $\lambda_1 = 0,$ the measurement is maximally unsharp $E_+^{\lambda_1} = E_-^{\lambda_1} = \mathbb{I}/2$, yielding no information and thus no violation in the first round, while enabling a maximal violation in the second round. This creates a trade-off between the possible violations between the two rounds, more violation in one leads to lesser violation in the other.

This behavior can be exploited. We can set a small value $\lambda_1 = \epsilon > 0$ arbitrarily close to zero and the maximum value $\lambda_2 = 1$. This ensures a finite violation in round one, $|\mathcal{B}_1| = 2\sqrt{1+4\alpha(1-\alpha)\epsilon^2} \approx 2+4\alpha(1-\alpha)\epsilon^2>2$, and a substantial violation in round two, $|\mathcal{B}_2| = 2\sqrt{\alpha(1-\alpha)[1+\sqrt{1-\epsilon^2}]^2 +1-\epsilon^2} \approx 2\sqrt{1+4\alpha(1-\alpha)} - \frac{1+2\alpha(1-\alpha)}{\sqrt{1+4\alpha(1-\alpha)}}\epsilon^2 >2$. For $\alpha=1/2$, the violation is near-maximal: $|\mathcal{B}_2| \approx 2\sqrt{2} - \frac{3}{2\sqrt{2}}\epsilon^2$. We conclude that unsharp measurements allow for near-maximal violations in round two while retaining at least some violation in round one; hence $k_{\max}\geq 2$. Now that we have proved the success of the random access Bell game by showing that $k_{\max}$ is indeed lower-bounded by two, we confront two important questions: 1) What is the value of $k_{\max}$? 2) What is the maximum possible violation in round $k$ for $k<k_{\max}$? In what follows, we use the patterns observed so far to answer these questions. 

\subsubsection{The Random Access Bell Game in the $k$-th Round}
\label{sec:k_round}

In this subsection, we deduce theorems on the behavior of the random access Bell game. We elucidate the general forms of the states and the achievable Bell violation in an arbitrary $k$-th round before tackling the questions posed earlier. Before proceeding with the general case, we introduce some notation 
\begin{equation}
R(k) = \prod_{i=1}^{k}\sqrt{1-\lambda_i^2} \quad \text{and} \quad S(k) = \prod_{i=1}^k\frac{1+\sqrt{1-\lambda_i^2}}{2}.
\end{equation}
It is evident that $R(k) = R(k-1)\sqrt{1-\lambda_k^2}$ and $S(k) = S(k-1)\frac{1+\sqrt{1-\lambda_k^2}}{2}$. We also write $\ket{\text{GHZ}_{\alpha+}} = \sqrt{\alpha}\ket{000}+\sqrt{1-\alpha}\ket{111}$ and $\ket{\text{GHZ}_{\alpha-}} = \sqrt{\alpha}\ket{000}-\sqrt{1-\alpha}\ket{111}$. Using these conventions, we propose the general forms of the quantum states in the $k$-th round of the protocol.

\begin{theorem}
The state $\rho_{CAB_{k}}$, obtained after running the protocol $k-1$ times and implementing the {\tt {\tt SWITCH}} action in the $k$-th round, has the form 
\begin{eqnarray}
\rho_{CAB_{k}} &=& \left[\frac{1+R(k-1)}{4} +\frac{S(k-1)}{2}\right]\ket{\text{GHZ}_{\alpha+}}\bra{\text{GHZ}_{\alpha+}} \notag\\
&+& \left[\frac{1+R(k-1)}{4} -\frac{S(k-1)}{2}\right]\ket{\text{GHZ}_{\alpha-}}\bra{\text{GHZ}_{\alpha-}} \notag\\
&+& \left[\frac{1-R(k-1)}{4}\right] \bigg(2(1-\alpha)\ket{010}\bra{010}\notag\\
&& \hspace{4cm} + 2\alpha\ket{101}\bra{101}\bigg)
\label{eqn:rhocab}
\end{eqnarray}
\label{Th:rhocab_general}
\end{theorem}
\begin{proof}
We prove this by induction. From Eqs.~\eqref{eqn:rhocab1} and~\eqref{eqn:rhocab2}, it is easy to verify that $\rho_{CAB_{1}}$ and $\rho_{CAB_{2}}$ satisfy the above form. Assuming it is also true for some $\rho_{CAB_{k}}$, the induction step for  $\rho_{CAB_{k+1}}$ can be derived by proceeding with the protocol. We prove the induction step in Appendix \ref{app:proof_th1}. 
\end{proof}

The state $\rho_{CAB_{k}}$ is independent of which of the two outcomes was produced during the unsharp measurement at any of the preceding $k-1$ rounds. This is because, just as in the first two rounds, the {\tt {\tt SWITCH}} action maps possible post-measurement states to the same output state, as can be inferred from the proof in Appendix \ref{app:proof_th1}. Using this form of the state, the post-measurement state obtained during the $k$-th round can be determined. The form is derived in Appendix \ref{app:proof_th1} (Eq.~\eqref{eq:appendix_postmeasurement_state}).

The state shared between Alice and Bob after the unsharp measurement can be calculated from Eq.~\eqref{eq:appendix_postmeasurement_state} by simply tracing out the subsystem C. It is also evident that the derived form is a generalization of the states obtained in the first two rounds (Eqs.~\eqref{eqn:rhoEcab1} and~\eqref{eqn:rhoEcab2}).    

\begin{corollary}
The post-measurement state of the subsystem shared by Alice and Bob at the $k$-th round, $\rho_{AB_k}^{E_\pm}$, has the form 
\begin{eqnarray}
&&\rho_{AB_{k}}^{E_\pm} \notag\\ &&= \left[\frac{1+R(k-1)}{4} +\frac{S(k-1)}{2}\right]\left(\frac{1\pm \lambda_k}{2} \ket{\Phi^+_\alpha}\bra{\Phi^+_\alpha} \right.\notag\\ && \hspace{5cm}\left.+\frac{1\mp \lambda_k}{2} \ket{\Phi^-_\alpha}\bra{\Phi^-_\alpha} \right) \notag\\
&&+ \left[\frac{1+R(k-1)}{4} -\frac{S(k-1)}{2}\right]\left(\frac{1\pm \lambda_k}{2} \ket{\Phi^-_\alpha}\bra{\Phi^-_\alpha} \right.\notag\\ && \hspace{5cm} \left.+\frac{1\mp \lambda_k}{2} \ket{\Phi^+_\alpha}\bra{\Phi^+_\alpha} \right) \notag \\
&&+ \left[\frac{1-R(k-1)}{4}\right] \bigg(\ket{\Psi^+_\alpha}\bra{\Psi^+_\alpha} + \ket{\Psi^-_\alpha}\bra{\Psi^-_\alpha}\bigg).
\label{eqn:rhoEab}
\end{eqnarray}
\end{corollary}

We can now evaluate the $T_{\rho}$ matrix to measure the Bell violation on the A:B cut. Given that the state is always in the Bell basis, the form of the matrix is particularly convenient.  

\begin{theorem}
The matrix $T_{\rho_{AB_k}^{E_\pm}}$ for the post-measurement subsystem of Alice and Bob during the $k$-th round of the protocol has the form
\begin{align}
T_{\rho_{AB_k}^{E_\pm}}
&= \mathrm{\text{Diag}}\Big(
     \pm 2\sqrt{\alpha(1-\alpha)}\,S(k-1)\lambda_k, \notag\\
&\hspace{1cm}
     \mp 2\sqrt{\alpha(1-\alpha)}\,S(k-1)\lambda_k, R(k-1)
   \Big) \notag\\[4pt]
&= \mathrm{\text{Diag}}\Big(
     \pm 2\sqrt{\alpha(1-\alpha)}
       \prod_{i=1}^{k-1}\tfrac{1+\sqrt{1-\lambda_i^2}}{2}\lambda_k, \notag\\
&\hspace{1cm}
     \mp 2\sqrt{\alpha(1-\alpha)}
       \prod_{i=1}^{k-1}\tfrac{1+\sqrt{1-\lambda_i^2}}{2}\lambda_k, 
     \prod_{i=1}^{k-1}\sqrt{1-\lambda_i^2}
   \Big).
\label{eqn:Tmatrix}
\end{align}

where the signs $\pm$ depend upon the measurement outcome in the $k$-th round.
\label{Th:Tmatrix}
\end{theorem}
\begin{proof}
The proof is given in Appendix \ref{app:proof_th2}.
\end{proof}

The matrix has particularly simple form in terms of the functions defined earlier. It is a direct generalization of the forms obtained in the first two rounds. Using the $T_\rho$ matrix, it is straightforward to calculate the Bell violation after $k$ rounds. This leads to our more important result.

\begin{theorem}
The Bell violation $|\mathcal{B}_k|$ after $k$ rounds is given by
\begin{align}
|\mathcal{B}_k| 
&= \max\Bigg(
    2\sqrt{2}\,\sqrt{4\alpha(1-\alpha)}\,
    S(k\!-\!1)\lambda_k,\notag\\
&\qquad\quad
    2\sqrt{
        4\alpha(1-\alpha)\,S(k\!-\!1)^2\lambda_k^2
        + R(k\!-\!1)^2
    }
\Bigg) \notag\\[4pt]
&\hspace{-0.7cm}=  \max\left(
    2\sqrt{2}\,\sqrt{4\alpha(1-\alpha)}\,
    \Bigg[\prod_{i=1}^{k-1}
        \tfrac{1+\sqrt{1-\lambda_i^2}}{2}
    \Bigg]\lambda_k,\right.\notag\\
&\left.
    2\sqrt{
        4\alpha(1-\alpha)
        \Bigg[\prod_{i=1}^{k-1}
            \tfrac{1+\sqrt{1-\lambda_i^2}}{2}
        \Bigg]^2\lambda_k^2
        + \Bigg[\prod_{i=1}^{k-1}(1-\lambda_i^2)\Bigg]
    }
\right)
\label{eqn:Bellviolation_wrt_lambdaparams}
\end{align}
\end{theorem}
\begin{proof}
Starting with \begin{align}
T_{\rho_{AB_k}^{E_\pm}}
&= \mathrm{\text{Diag}}\Big(
     \pm 2\sqrt{\alpha(1-\alpha)}\,S(k-1)\lambda_k, \notag\\
&\hspace{1cm}
     \mp 2\sqrt{\alpha(1-\alpha)}\,S(k-1)\lambda_k, R(k-1)
   \Big), \notag
\end{align} the $U_{\rho_{AB_k}^{E_\pm}}$ matrix can be calculated as 
\begin{eqnarray}
U_{\rho_{AB_k}^{E_\pm}}&=&\Big(T_{\rho_{AB_k}^{E_\pm}}\Big)^T\Big(T_{\rho_{AB_k}^{E_\pm}}\Big) \notag\\
&=& \text{Diag}\left( 4\alpha(1-\alpha) S(k-1)^2\lambda_k^2, \right.\notag\\
&&\left.~ 4\alpha(1-\alpha)S(k-1)^2\lambda_k^2,~ R(k-1)^2\right).\notag \\
\end{eqnarray}
Now, the critical quantity $M\left(\rho_{AB_k}^{E_\pm}\right)$ is the sum of the two largest eigenvalues of the diagonal $U_{\rho_{AB_k}^{E_\pm}}$ matrix. Since, $|\mathcal{B}_k| = 2\sqrt{M(\rho_{AB_k}^{E_\pm})}$, we get the desired expression. 
\end{proof}

Just like in the first two rounds, the Bell violation is independent of the chosen post-measurement outcomes of any of the unsharp measurements. All outcomes lead to the same Bell violation. The violation decays as $\alpha$ values get closer to 0 or 1 as expected.

Also, similar to the first two rounds,a trade-off is introduced by the unsharp measurements. The Bell violation in round $k$, $|\mathcal{B}_k|$, increases with the corresponding sharpness parameter $\lambda_k$ but decreases with the sharpness parameters $\{\lambda_1, \dots,\lambda_{k-1} \}$ of all preceding rounds. Consequently, a higher violation in an earlier round comes at the cost of the potential violation in a later one. This behavior can be exploited to generate desired Bell violations, which is the key to our next set of results, where we determine the value of $k_{\max}$ and the maximum achievable violation in any round $k\leq k_{\max}$.

\begin{theorem} A Bell violation can be obtained for an arbitrarily large number of rounds; thus, $k_{\max} = \infty$.
\label{Th:theorem_Bellviolation}
\end{theorem}

\begin{proof}
Suppose we wish to establish a Bell violation in round $k$ and all preceding rounds. The observed trade-off suggests the following strategy. We select a scaling parameter $q>1$ and set the sharpness parameters as $\lambda_k =1, \lambda_{k-1}=1/q, \dots, \lambda_m = 1/q^{k-m}, \dots, \lambda_1 = 1/q^{k-1}$. For any round $m$, the sharpness parameter $\lambda_m$ is significantly higher than those of all preceding rounds. This leads to a Bell violation in the $m$-th round of

\begin{widetext}
\begin{eqnarray}
|\mathcal{B}_m|&=& \max\Bigg( \frac{2\sqrt{2}}{q^{k-m}}\sqrt{4\alpha(1-\alpha)}\Bigg[\prod_{i=1}^{m-1}\frac{1+\sqrt{1-q^{-2(k-i)}}}{2}\Bigg], \notag\\
&& \hspace{3cm}~2\sqrt{4\alpha(1-\alpha)\Bigg[\prod_{i=1}^{m-1}\frac{1+\sqrt{1-q^{-2(k-i)}}}{2}\Bigg]^2 \frac{1}{q^{2(k-m)}} + \Bigg[\prod_{i=1}^{m-1}1-q^{-2(k-i)}\Bigg]}~\Bigg).
\label{eq:bellviol_m_sharpness}
\end{eqnarray}
\end{widetext}
To prove $|\mathcal{B}_m| > 2$, it suffices to show that the term under the square root in the second argument is greater than 1. In Appendix \ref{app:proof_th3}, we prove that by setting $q \geq \sqrt{2/[\alpha(1-\alpha)]}$, we have
\begin{eqnarray}
&& 4\alpha(1-\alpha)\left[\prod_{i=1}^{m-1}\frac{1+\sqrt{1-q^{-2(k-i)}}}{2}\right]^2 \frac{1}{q^{2(k-m)}} \notag\\ &&\hspace{2.5cm} + \prod_{i=1}^{m-1}(1-q^{-2(k-i)}) > 1,
\end{eqnarray}
and thus $|\mathcal{B}_m| > 2$, constituting a Bell violation at any arbitrary round $m \leq k$. Since $k$ was also chosen arbitrarily, this proves that a violation can be observed for an arbitrarily large number of rounds, and thus $k_{\max} = \infty$.
\end{proof}

This result also shows that generalized GHZ states are a powerful resource for our protocol, since any initial state with non-trivial entanglement $(\alpha \neq 0,1)$ can provide violations for an arbitrary number of rounds. As we increase the scaling factor $q$, the drop off in sharpness parameter for all the rounds before $k$ becomes much steeper. As the violation in round $k$ increases with a decrease in the sharpness parameters of previous rounds, we expect the violation $|\mathcal{B}_k|$ to grow with the scaling factor $q$. This is indeed the case, by tweaking the scaling parameter, and $\alpha$ value, Bell violation arbitrary close to the Tsirelson bound of $2\sqrt{2}$ can be obtained in the $k$th round,while still retaining at least some violation in all preceding rounds.

\begin{theorem} A near-maximal Bell violation can be obtained in the $k$-th round of a successful Bell game, where $k$ can be made arbitrarily large. 
\label{Th:near_max_viol}
\end{theorem}
\begin{proof}
Using the scaling factor $q$ and setting the sharpness parameters to $\{\lambda_k=1, \lambda_{k-1} = 1/q, \dots, \lambda_1=1/q^{k-1}\}$, along with $\alpha = 1/2$, we evaluate the violation observed in the final round $k$ from expression Eq.~\eqref{eq:bellviol_m_sharpness}:

\begin{eqnarray}
|\mathcal{B}_k| &=& \max\left( 2\sqrt{2}\prod_{i=1}^{k-1}\frac{1+\sqrt{1-q^{-2(k-i)}}}{2}, \right. \notag\\
&& \hspace{-1cm}\left. 2\sqrt{\left[\prod_{i=1}^{k-1}\frac{1+\sqrt{1-q^{-2(k-i)}}}{2}\right]^2 + \prod_{i=1}^{k-1}(1-q^{-2(k-i)})}\right). \notag\\ \end{eqnarray}
Since $\frac{1+\sqrt{1-q^{-2(k-i)}}}{2}$ is the average of 1 and $\sqrt{1-q^{-2(k-i)}}$, it is always larger than the $\sqrt{1-q^{-2(k-i)}}$ term and thus 
\begin{eqnarray}
&& |\mathcal{B}_k|=  2\sqrt{2}\Bigg[\prod_{i=1}^{k-1}\frac{1+\sqrt{1-q^{-2(k-i)}}}{2}\Bigg], 
\end{eqnarray}
It is evident that $|\mathcal{B}_k|$ grows with $q$. In the limit of arbitrarily large $q$ we get 
\begin{eqnarray}
 \lim_{q\to \infty}|\mathcal{B}_k| &\approx&  2\sqrt{2}\prod_{i=1}^{k-1}\Bigg[1-\frac{1}{4q^{2(k-i)}}\Bigg] \notag\\ &\approx& 2\sqrt{2} -\frac{1}{\sqrt{2}q^2} +O\left(\frac{1}{q^4}\right).  
\end{eqnarray}
This implies that we can set an arbitrarily high value for $q$ to obtain any desired violation up to the Tsirelson bound. The guarantee of at least some violation in all previous rounds is ensured by Theorem \ref{Th:theorem_Bellviolation}.
\end{proof}

We have proved not only that violations can be obtained for any number of rounds but also that a near-maximal violation can be achieved at any round while retaining at least some violation in all preceding rounds. This demonstrates that a combination of the quantum {\tt SWITCH} and unsharp measurements can be utilized effectively to preserve Bell violations against entanglement-breaking channels. Furthermore, as all post measurement outcomes lead to the same states and violations and successive rounds, we could set in our protocol that no classical information is transmitted between rounds.

While it is evident that at least some Bell violation is possible in as many rounds as needed, the trade-off behavior does put some constraints on the total available violation. For any given $\alpha$, let us set a minimum threshold on the Bell violation value $|\mathcal{B}_{\min}|$. Now, we require the players to tweak the $\lambda$ parameters in such a manner to attain at least $|\mathcal{B}_{\min}|$ at every round. From our analysis, we know that a successful violation at some round $k$ requires the sharpness parameter of that round $\lambda_k$ to be significantly higher than all the preceding rounds. At the same time the maximum $\lambda_k$ value at any round can not go beyond 1. This suggests that with a fixed $|\mathcal{B}_{\min}|$ and $\alpha$, there should only be a finite number of maximum rounds $N_{\max}(|\mathcal{B}_{\min}|, \alpha)$ such that a violation of at least $|\mathcal{B}_{\min}|$ is attained in every round up till $N_{\max}(|\mathcal{B}_{\min}|, \alpha)$. With the analytical expression of $|\mathcal{B}_k|$ with respect to $\lambda$ parameters Eq.~\eqref{eqn:Bellviolation_wrt_lambdaparams} available. We can easily evaluate  $N_{\max}(|\mathcal{B}_{\min}|, \alpha)$ numerically by choosing minimum $\lambda_k$ in successive rounds so that $|\mathcal{B}_{\min}|$ is obtained. We plot $N_{\max}(|\mathcal{B}_{\min}|, \alpha)$ to access this behavior.  

\begin{figure}[ht]
    \centering
\includegraphics[width=\linewidth]{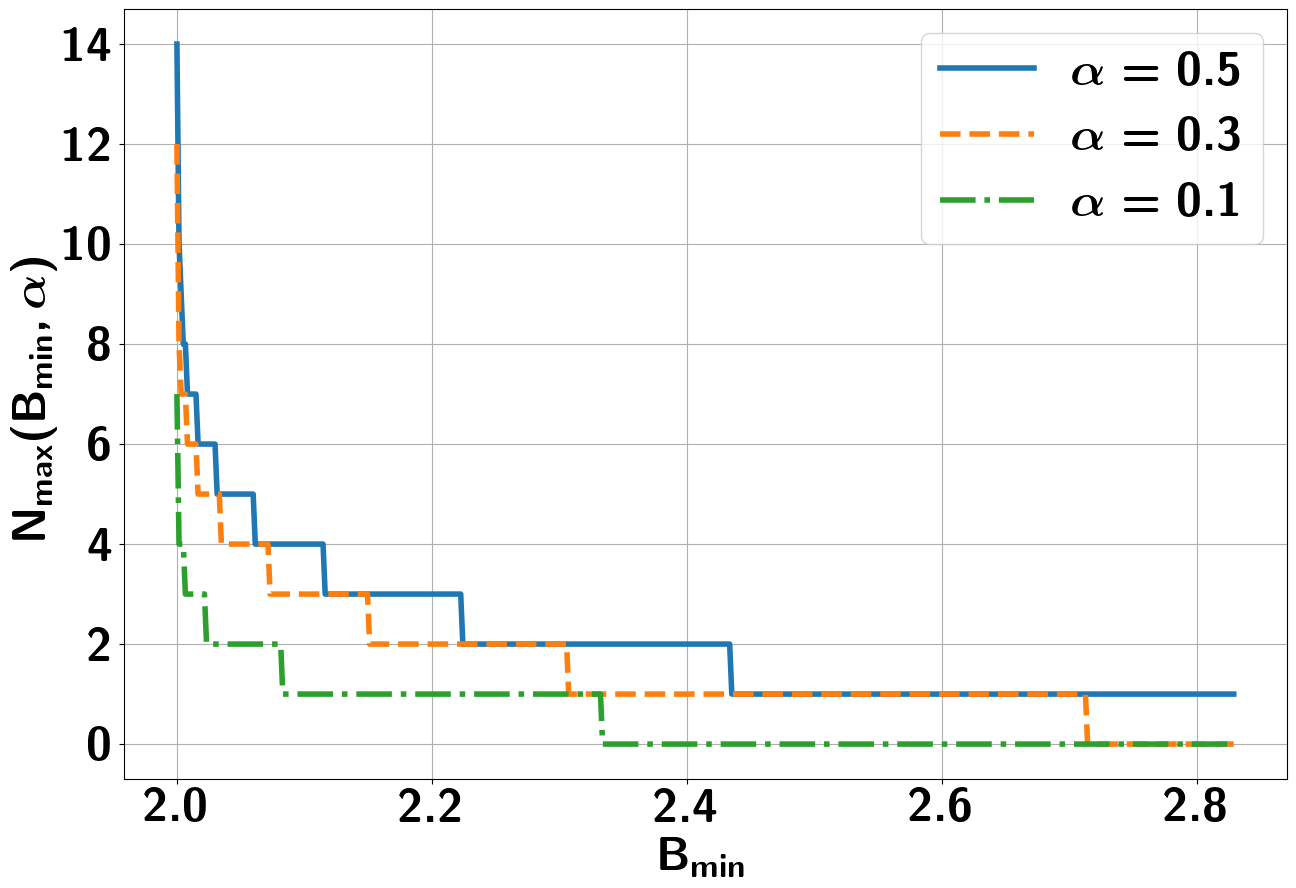}
    \caption{\textbf{Persistence of Bell Violation}. Maximum rounds of violation $N_{\max}(|\mathcal{B}_{\min}|, \alpha)$ (ordinate) with respect to violation threshold $|\mathcal{B}_{\min}|$. The $N_{\max}$ value approaches $\infty$ as the threshold $|\mathcal{B}_{\min}|$ approaches 2, as predicted by our analysis. All axes are dimensionless.}
    \label{fig:Rmax}
\end{figure}

Fig. \ref{fig:Rmax} demonstrates how $N_{\max}$ scales with $|\mathcal{B}_{\min}|$. Higher the $|\mathcal{B}_{\min}|$, the fewer the rounds in which this level of violation can be observed. We also see that $N_{\max}$ increases as $\alpha$ approaches $1/2$ as expected. While, $N_{\max}$ value is nominal for higher threshold values, it tends towards infinity as the threshold tends to 2.

\subsection{Random Access Bell Game with a W State}
\label{sec:w_state}

Our analysis has shown that a generalized GHZ state is a powerful resource for the proposed Bell game. It is instructive, therefore, to consider the case where the initial state is a W state, $\ket{W} = \frac{1}{\sqrt{3}}(\ket{001}+\ket{010}+\ket{100})$, $\rho_{CAB_{0}} = \ketbra{W}{W}$ instead of a GHZ state, while keeping the rest of the protocol the same. We find that with a W state, no Bell violation is possible even in the first round. This shows that the GHZ type entanglement is pertinent to our protocol.

\begin{theorem}
The random access Bell game fails for W states $\rho_{CAB_{0}} = \ketbra{W}{W}$, leading to $k_{\max} = 0$.
\end{theorem}
\begin{proof}
Starting with the W state, the state $\rho_{CAB_{1}}$ obtained after running the {\tt {\tt SWITCH}} action in the first round results in the state
\begin{eqnarray}
\rho_{CAB_{1}} &=& \frac{1}{3}(\ket{000}\bra{000}+\ket{101}\bra{101} +\ket{010}\bra{010}). \notag\\
\end{eqnarray}

A subsequent measurement on the control system C in the $\ket{\pm}$ basis yields 
\begin{eqnarray}
\rho_{CAB_{1}}^{E_\pm} &=& \frac{1\pm 1}{6}\ket{+00}\bra{+00}+\frac{1\mp 1}{6}\ket{-00}\bra{-00}
\notag\\ && +\frac{1\pm 1}{6}\ket{+01}\bra{+01}+\frac{1\mp 1}{6}\ket{-01}\bra{-01} \notag\\
&&+ \frac{1\pm 1}{6}\ket{+10}\bra{+10}+\frac{1\mp 1}{6}\ket{-10}\bra{-10}, \notag\\ 
\end{eqnarray}
which leads to the following joint state with Alice and Bob on the subsystem AB 
 \begin{eqnarray}
\rho_{AB_{1}}^{E_\pm} &=& \frac{1}{3}\left(\ket{00}\bra{00} + \ket{01}\bra{01}+\ket{10}\bra{10}\right) \notag\\
\end{eqnarray}
This is a separable state with respect to the A:B partition and therefore cannot violate a Bell inequality. Since the protocol fails to provide a violation in the first round, we conclude that the W state is unsuitable for this game, and thus $k_{\max} = 0$. A proof for the general case of an arbitrary round $k$ is provided in Appendix \ref{app:proof_th5}.
\end{proof}

We conclude that generalized GHZ state or GHZ type entanglement is an important resource for our protocol. While any generalised GHZ state with $\alpha \in (0,1)$  can provide Bell violations for an arbitrary number of rounds, the W state is unable to provide a violation in any round.\\

\section{Conclusion}
\label{sec:conclusion}
In this paper, we introduced the Random Access Bell Game as a task for studying the resilience of nonlocality against repeated environmental interactions. We analyzed a scenario where the noise is composed of two distinct entanglement-breaking pin maps (complete erasure channels) and demonstrated that conventional strategies, including probabilistic compositions or a canonical quantum {\tt SWITCH} setup, are insufficient to preserve nonlocality. Our main contribution is a protocol that synergistically combines three resources: the decoherence-free subspace created by the quantum {\tt SWITCH}, the controllable disturbance of sequential unsharp measurements, and the specific correlation structure of GHZ-type entanglement. We proved that our protocol not only succeeds in the game but can preserve the ability to violate the CHSH inequality for an arbitrarily large number of rounds, effectively achieving $k_{\text{max}} = \infty$. This result offers a strategy for mitigating environmental noise, demonstrating that quantum correlations can, in principle, be sustained indefinitely despite repeated interactions with entanglement-breaking environments.
%demonstrating that quantum correlations can be indefinitely sustained even through repeated encounters with entanglement-breaking environments.

%The work establishes a novel method for protecting quantum resources. 
Instead of correcting errors after they occur, our protocol utilizes the form of the system-environment interaction itself via indefinite causal order to create a decoherence-free subspace. We then use generalized measurements to manage the trade-off between information gain and disturbance to obtain success in the Random Access Bell Game. 
%This provides a valuable strategy against certain types of structured noise that circumvents the resource overhead of traditional approaches of shielding a system from noise or reactively correcting errors after they occur. 
 We have also shown that by strategically choosing the measurement sharpness parameters, one can achieve a near-maximal Bell violation at arbitrarily chosen round $k$, while still guaranteeing success in all preceding rounds. Interestingly, the protocol fails for W-states, thereby establishing the Random Access Bell Game as an operational task that can distinguish between GHZ and W states.
%This demonstrates a remarkable level of dynamic control over the nonlocal resources. 
%The established failure of the protocol for W-states demonstrates the specificity of our method and also provides a clear and practical operational distinction between the GHZ and W classes of tripartite entanglement in the context of noise resilience.

The findings open several avenues for future research. A natural next step is to investigate the protocol's performance against other, more general noise models, such as partial erasure channels, depolarizing, dephasing, or amplitude damping channels.
%to determine if similar protected subspaces arises for other noise models as well. 
The framework itself can be extended to protect other quantum resources beyond nonlocality, such as quantum steering or coherence, in similar sequential scenarios. Another possible route for generalization is to construct the game where each noisy block contains more than two channels, or an extension in a multipartite setting, involving more than two parties, possibly having higher dimensions.
%could yield new strategies for preserving multipartite nonlocality against collective noise. 
Finally, given that both the quantum {\tt SWITCH} and unsharp measurements have been realized experimentally, assessing the feasibility of implementing this sequential protocol in a laboratory setting presents a tangible yet exciting challenge.\\

\section*{Acknowledgment}
The authors acknowledge discussions with Pritam Halder and Tamal Guha.
SR acknowledges discussions with Arkaprabha Ghoshal and Pratik Ghosal at YouQu 2023, HRI, Prayagraj, India.

\bibliography{bib}

@article{Popescu2014Bell,
  title = {Bell's inequalities versus teleportation: What is nonlocality?},
  author = {Popescu, Sandu},
  journal = {Phys. Rev. Lett.},
  volume = {72},
  issue = {6},
  pages = {797--799},
  numpages = {0},
  year = {1994},
  month = {Feb},
  publisher = {American Physical Society},
  doi = {10.1103/PhysRevLett.72.797},
  url = {https://link.aps.org/doi/10.1103/PhysRevLett.72.797}
}

@article{Buscemi2014All,
  title = {All Entangled Quantum States Are Nonlocal},
  author = {Buscemi, Francesco},
  journal = {Phys. Rev. Lett.},
  volume = {108},
  issue = {20},
  pages = {200401},
  numpages = {5},
  year = {2012},
  month = {May},
  publisher = {American Physical Society},
  doi = {10.1103/PhysRevLett.108.200401},
  url = {https://link.aps.org/doi/10.1103/PhysRevLett.108.200401}
}

@article{Soorya2019Universality,
  title = {Universality in distribution of monogamy scores for random multiqubit pure states},
  author = {Rethinasamy, Soorya and Roy, Saptarshi and Chanda, Titas and Sen(De), Aditi and Sen, Ujjwal},
  journal = {Phys. Rev. A},
  volume = {99},
  issue = {4},
  pages = {042302},
  numpages = {11},
  year = {2019},
  month = {Apr},
  publisher = {American Physical Society},
  doi = {10.1103/PhysRevA.99.042302},
  url = {https://link.aps.org/doi/10.1103/PhysRevA.99.042302}
}

@article{Bell1964on,
  title = {On the Einstein Podolsky Rosen paradox},
  author = {Bell, J. S.},
  journal = {Physics Physique Fizika},
  volume = {1},
  issue = {3},
  pages = {195--200},
  numpages = {6},
  year = {1964},
  month = {Nov},
  publisher = {American Physical Society},
  doi = {10.1103/PhysicsPhysiqueFizika.1.195},
  url = {https://link.aps.org/doi/10.1103/PhysicsPhysiqueFizika.1.195}
}

@article{Cirelson1980,
  title = {Quantum generalizations of Bell’s inequality},
  volume = {4},
  ISSN = {1573-0530},
  url = {http://dx.doi.org/10.1007/BF00417500},
  DOI = {10.1007/bf00417500},
  number = {2},
  journal = {Letters in Mathematical Physics},
  publisher = {Springer Science and Business Media LLC},
  author = {Cirel’son,  B. S.},
  year = {1980},
  month = mar,
  pages = {93–100}
}

@article{navascas2007,
  title = {Bounding the Set of Quantum Correlations},
  author = {Navascu\'es, Miguel and Pironio, Stefano and Ac\'{\i}n, Antonio},
  journal = {Phys. Rev. Lett.},
  volume = {98},
  issue = {1},
  pages = {010401},
  numpages = {4},
  year = {2007},
  month = {Jan},
  publisher = {American Physical Society},
  doi = {10.1103/PhysRevLett.98.010401},
  url = {https://link.aps.org/doi/10.1103/PhysRevLett.98.010401}
}

@article{Barizien2025,
  title = {Quantum statistics in the minimal Bell scenario},
  volume = {21},
  ISSN = {1745-2481},
  url = {http://dx.doi.org/10.1038/s41567-025-02782-3},
  DOI = {10.1038/s41567-025-02782-3},
  number = {4},
  journal = {Nature Physics},
  publisher = {Springer Science and Business Media LLC},
  author = {Barizien,  Victor and Bancal,  Jean-Daniel},
  year = {2025},
  month = mar,
  pages = {577–582}
}

@article{Le2025,
  title = {The limits of quantum correlations},
  volume = {21},
  ISSN = {1745-2481},
  url = {http://dx.doi.org/10.1038/s41567-025-02836-6},
  DOI = {10.1038/s41567-025-02836-6},
  number = {4},
  journal = {Nature Physics},
  publisher = {Springer Science and Business Media LLC},
  author = {Le,  Thinh P.},
  year = {2025},
  month = mar,
  pages = {501–502}
}

@article{Hensen2015,
  title = {Loophole-free Bell inequality violation using electron spins separated by 1.3 kilometres},
  volume = {526},
  ISSN = {1476-4687},
  url = {http://dx.doi.org/10.1038/nature15759},
  DOI = {10.1038/nature15759},
  number = {7575},
  journal = {Nature},
  publisher = {Springer Science and Business Media LLC},
  author = {Hensen,  B. and Bernien,  H. and Dréau,  A. E. and Reiserer,  A. and Kalb,  N. and Blok,  M. S. and Ruitenberg,  J. and Vermeulen,  R. F. L. and Schouten,  R. N. and Abellán,  C. and Amaya,  W. and Pruneri,  V. and Mitchell,  M. W. and Markham,  M. and Twitchen,  D. J. and Elkouss,  D. and Wehner,  S. and Taminiau,  T. H. and Hanson,  R.},
  year = {2015},
  month = oct,
  pages = {682–686}
}

@article{Giustina2015,
  title = {Significant-Loophole-Free Test of Bell's Theorem with Entangled Photons},
  author = {Giustina, Marissa and Versteegh, Marijn A. M. and Wengerowsky, S\"oren and Handsteiner, Johannes and Hochrainer, Armin and Phelan, Kevin and Steinlechner, Fabian and Kofler, Johannes and Larsson, Jan-\AA{}ke and Abell\'an, Carlos and Amaya, Waldimar and Pruneri, Valerio and Mitchell, Morgan W. and Beyer, J\"orn and Gerrits, Thomas and Lita, Adriana E. and Shalm, Lynden K. and Nam, Sae Woo and Scheidl, Thomas and Ursin, Rupert and Wittmann, Bernhard and Zeilinger, Anton},
  journal = {Phys. Rev. Lett.},
  volume = {115},
  issue = {25},
  pages = {250401},
  numpages = {7},
  year = {2015},
  month = {Dec},
  publisher = {American Physical Society},
  doi = {10.1103/PhysRevLett.115.250401},
  url = {https://link.aps.org/doi/10.1103/PhysRevLett.115.250401}
}

@article{loop,
  title = {Strong Loophole-Free Test of Local Realism},
  author = {Shalm, Lynden K. and Meyer-Scott, Evan and Christensen, Bradley G. and Bierhorst, Peter and Wayne, Michael A. and Stevens, Martin J. and Gerrits, Thomas and Glancy, Scott and Hamel, Deny R. and Allman, Michael S. and Coakley, Kevin J. and Dyer, Shellee D. and Hodge, Carson and Lita, Adriana E. and Verma, Varun B. and Lambrocco, Camilla and Tortorici, Edward and Migdall, Alan L. and Zhang, Yanbao and Kumor, Daniel R. and Farr, William H. and Marsili, Francesco and Shaw, Matthew D. and Stern, Jeffrey A. and Abell\'an, Carlos and Amaya, Waldimar and Pruneri, Valerio and Jennewein, Thomas and Mitchell, Morgan W. and Kwiat, Paul G. and Bienfang, Joshua C. and Mirin, Richard P. and Knill, Emanuel and Nam, Sae Woo},
  journal = {Phys. Rev. Lett.},
  volume = {115},
  issue = {25},
  pages = {250402},
  numpages = {10},
  year = {2015},
  month = {Dec},
  publisher = {American Physical Society},
  doi = {10.1103/PhysRevLett.115.250402},
  url = {https://link.aps.org/doi/10.1103/PhysRevLett.115.250402}
}

@article{Clauser1972,
  title = {Experimental Test of Local Hidden-Variable Theories},
  author = {Freedman, Stuart J. and Clauser, John F.},
  journal = {Phys. Rev. Lett.},
  volume = {28},
  issue = {14},
  pages = {938--941},
  numpages = {0},
  year = {1972},
  month = {Apr},
  publisher = {American Physical Society},
  doi = {10.1103/PhysRevLett.28.938},
  url = {https://link.aps.org/doi/10.1103/PhysRevLett.28.938}
}

@article{bexp1,
  title = {Experimental Tests of Realistic Local Theories via Bell's Theorem},
  author = {Aspect, Alain and Grangier, Philippe and Roger, G\'erard},
  journal = {Phys. Rev. Lett.},
  volume = {47},
  issue = {7},
  pages = {460--463},
  numpages = {0},
  year = {1981},
  month = {Aug},
  publisher = {American Physical Society},
  doi = {10.1103/PhysRevLett.47.460},
  url = {https://link.aps.org/doi/10.1103/PhysRevLett.47.460}
}

@article{bexp2,
  title = {Experimental Test of Bell's Inequalities Using Time-Varying Analyzers},
  author = {Aspect, Alain and Dalibard, Jean and Roger, G\'erard},
  journal = {Phys. Rev. Lett.},
  volume = {49},
  issue = {25},
  pages = {1804--1807},
  numpages = {0},
  year = {1982},
  month = {Dec},
  publisher = {American Physical Society},
  doi = {10.1103/PhysRevLett.49.1804},
  url = {https://link.aps.org/doi/10.1103/PhysRevLett.49.1804}
}

@article{bexp3,
  title = {Violation of Bell Inequalities by Photons More Than 10 km Apart},
  author = {Tittel, W. and Brendel, J. and Zbinden, H. and Gisin, N.},
  journal = {Phys. Rev. Lett.},
  volume = {81},
  issue = {17},
  pages = {3563--3566},
  numpages = {0},
  year = {1998},
  month = {Oct},
  publisher = {American Physical Society},
  doi = {10.1103/PhysRevLett.81.3563},
  url = {https://link.aps.org/doi/10.1103/PhysRevLett.81.3563}
}

@article{diqkd1,
  title = {Device-Independent Quantum Key Distribution with Local Bell Test},
  author = {Lim, Charles Ci Wen and Portmann, Christopher and Tomamichel, Marco and Renner, Renato and Gisin, Nicolas},
  journal = {Phys. Rev. X},
  volume = {3},
  issue = {3},
  pages = {031006},
  numpages = {11},
  year = {2013},
  month = {Jul},
  publisher = {American Physical Society},
  doi = {10.1103/PhysRevX.3.031006},
  url = {https://link.aps.org/doi/10.1103/PhysRevX.3.031006}
}

@article{diqkd2,
  title = {Advances in device-independent quantum key distribution},
  volume = {9},
  ISSN = {2056-6387},
  url = {http://dx.doi.org/10.1038/s41534-023-00684-x},
  DOI = {10.1038/s41534-023-00684-x},
  number = {1},
  journal = {npj Quantum Information},
  publisher = {Springer Science and Business Media LLC},
  author = {Zapatero,  Víctor and van Leent,  Tim and Arnon-Friedman,  Rotem and Liu,  Wen-Zhao and Zhang,  Qiang and Weinfurter,  Harald and Curty,  Marcos},
  year = {2023},
  month = feb 
}

@article{diqkd3,
  title = {Device-Independent Quantum Key Distribution with Arbitrarily Small Nonlocality},
  author = {Wooltorton, Lewis and Brown, Peter and Colbeck, Roger},
  journal = {Phys. Rev. Lett.},
  volume = {132},
  issue = {21},
  pages = {210802},
  numpages = {6},
  year = {2024},
  month = {May},
  publisher = {American Physical Society},
  doi = {10.1103/PhysRevLett.132.210802},
  url = {https://link.aps.org/doi/10.1103/PhysRevLett.132.210802}
}

@article{diqkd4,
  title = {Device-Independent Quantum Key Distribution Based on Routed Bell Tests},
  author = {Le Roy-Deloison, Tristan and Lobo, Edwin Peter and Pauwels, Jef and Pironio, Stefano},
  journal = {PRX Quantum},
  volume = {6},
  issue = {2},
  pages = {020311},
  numpages = {15},
  year = {2025},
  month = {Apr},
  publisher = {American Physical Society},
  doi = {10.1103/PRXQuantum.6.020311},
  url = {https://link.aps.org/doi/10.1103/PRXQuantum.6.020311}
}

@article{Pironio2010,
  title = {Random numbers certified by Bell’s theorem},
  volume = {464},
  ISSN = {1476-4687},
  url = {http://dx.doi.org/10.1038/nature09008},
  DOI = {10.1038/nature09008},
  number = {7291},
  journal = {Nature},
  publisher = {Springer Science and Business Media LLC},
  author = {Pironio,  S. and Acín,  A. and Massar,  S. and de la Giroday,  A. Boyer and Matsukevich,  D. N. and Maunz,  P. and Olmschenk,  S. and Hayes,  D. and Luo,  L. and Manning,  T. A. and Monroe,  C.},
  year = {2010},
  month = apr,
  pages = {1021–1024}
}

@article{Tan2016,
  title = {Biased Random Number Generator Based on Bell’s Theorem},
  volume = {33},
  ISSN = {1741-3540},
  url = {http://dx.doi.org/10.1088/0256-307X/33/3/030302},
  DOI = {10.1088/0256-307x/33/3/030302},
  number = {3},
  journal = {Chinese Physics Letters},
  publisher = {IOP Publishing},
  author = {Tan,  Yong-Gang and Hu,  Yao-Hua and Yang,  Hai-Feng},
  year = {2016},
  month = mar,
  pages = {030302}
}

@article{comcomplexity1,
  title = {Quantum Communication Complexity Protocol with Two Entangled Qutrits},
  author = {Brukner, \ifmmode \check{C}\else \v{C}\fi{}aslav and \ifmmode \dot{Z}\else \.{Z}\fi{}ukowski, Marek and Zeilinger, Anton},
  journal = {Phys. Rev. Lett.},
  volume = {89},
  issue = {19},
  pages = {197901},
  numpages = {4},
  year = {2002},
  month = {Oct},
  publisher = {American Physical Society},
  doi = {10.1103/PhysRevLett.89.197901},
  url = {https://link.aps.org/doi/10.1103/PhysRevLett.89.197901}
}

@article{comcomplexity2,
  title = {Bell's Inequalities and Quantum Communication Complexity},
  author = {Brukner, \ifmmode \check{C}\else \v{C}\fi{}aslav and \ifmmode \dot{Z}\else \.{Z}\fi{}ukowski, Marek and Pan, Jian-Wei and Zeilinger, Anton},
  journal = {Phys. Rev. Lett.},
  volume = {92},
  issue = {12},
  pages = {127901},
  numpages = {4},
  year = {2004},
  month = {Mar},
  publisher = {American Physical Society},
  doi = {10.1103/PhysRevLett.92.127901},
  url = {https://link.aps.org/doi/10.1103/PhysRevLett.92.127901}
}

@article{comcomplexity3,
  title = {Nonlocality and communication complexity},
  author = {Buhrman, Harry and Cleve, Richard and Massar, Serge and de Wolf, Ronald},
  journal = {Rev. Mod. Phys.},
  volume = {82},
  issue = {1},
  pages = {665--698},
  numpages = {0},
  year = {2010},
  month = {Mar},
  publisher = {American Physical Society},
  doi = {10.1103/RevModPhys.82.665},
  url = {https://link.aps.org/doi/10.1103/RevModPhys.82.665}
}

@inbook{Walls1994,
  title = {Bells Inequalities in Quantum Optics},
  ISBN = {9783642795046},
  url = {http://dx.doi.org/10.1007/978-3-642-79504-6_14},
  DOI = {10.1007/978-3-642-79504-6_14},
  booktitle = {Quantum Optics},
  publisher = {Springer Berlin Heidelberg},
  author = {Walls,  D. F. and Milburn,  G. J.},
  year = {1994},
  pages = {261–279}
}

@article{Chen2002,
  title = {Maximal Violation of Bell's Inequalities for Continuous Variable Systems},
  author = {Chen, Zeng-Bing and Pan, Jian-Wei and Hou, Guang and Zhang, Yong-De},
  journal = {Phys. Rev. Lett.},
  volume = {88},
  issue = {4},
  pages = {040406},
  numpages = {4},
  year = {2002},
  month = {Jan},
  publisher = {American Physical Society},
  doi = {10.1103/PhysRevLett.88.040406},
  url = {https://link.aps.org/doi/10.1103/PhysRevLett.88.040406}
}

@article{roy2018,
  title = {Response in the violation of the Bell inequality to imperfect photon addition and subtraction in noisy squeezed states of light},
  author = {Roy, Saptarshi and Chanda, Titas and Das, Tamoghna and Sen(De), Aditi and Sen, Ujjwal},
  journal = {Phys. Rev. A},
  volume = {98},
  issue = {5},
  pages = {052131},
  numpages = {16},
  year = {2018},
  month = {Nov},
  publisher = {American Physical Society},
  doi = {10.1103/PhysRevA.98.052131},
  url = {https://link.aps.org/doi/10.1103/PhysRevA.98.052131}
}

@article{Steinacker2025,
  title = {Bell inequality violation in gate-defined quantum dots},
  volume = {16},
  ISSN = {2041-1723},
  url = {http://dx.doi.org/10.1038/s41467-025-57987-0},
  DOI = {10.1038/s41467-025-57987-0},
  number = {1},
  journal = {Nature Communications},
  publisher = {Springer Science and Business Media LLC},
  author = {Steinacker,  Paul and Tanttu,  Tuomo and Lim,  Wee Han and Dumoulin Stuyck,  Nard and Feng,  MengKe and Serrano,  Santiago and Vahapoglu,  Ensar and Su,  Rocky Y. and Huang,  Jonathan Y. and Jones,  Cameron and Itoh,  Kohei M. and Hudson,  Fay E. and Escott,  Christopher C. and Morello,  Andrea and Saraiva,  Andre and Yang,  Chih Hwan and Dzurak,  Andrew S. and Laucht,  Arne},
  year = {2025},
  month = apr 
}

@article{Wang2002,
  title = {Quantum entanglement and Bell inequalities in Heisenberg spin chains},
  volume = {301},
  ISSN = {0375-9601},
  url = {http://dx.doi.org/10.1016/S0375-9601(02)00885-X},
  DOI = {10.1016/s0375-9601(02)00885-x},
  number = {1–2},
  journal = {Physics Letters A},
  publisher = {Elsevier BV},
  author = {Wang,  Xiaoguang and Zanardi,  Paolo},
  year = {2002},
  month = aug,
  pages = {1–6}
}

@article{Lee2022,
  title = {Detection of a quantum phase transition in a spin-1 chain through multipartite high-order correlations},
  author = {Lee, Dongkeun and Sohbi, Adel and Son, Wonmin},
  journal = {Phys. Rev. A},
  volume = {106},
  issue = {4},
  pages = {042432},
  numpages = {8},
  year = {2022},
  month = {Oct},
  publisher = {American Physical Society},
  doi = {10.1103/PhysRevA.106.042432},
  url = {https://link.aps.org/doi/10.1103/PhysRevA.106.042432}
}

@article{Getelina2018,
  title = {Violation of the Bell inequality in quantum critical random spin-1/2 chains},
  volume = {382},
  ISSN = {0375-9601},
  url = {http://dx.doi.org/10.1016/j.physleta.2018.08.003},
  DOI = {10.1016/j.physleta.2018.08.003},
  number = {39},
  journal = {Physics Letters A},
  publisher = {Elsevier BV},
  author = {Getelina,  João C. and de Oliveira,  Thiago R. and Hoyos,  José A.},
  year = {2018},
  month = oct,
  pages = {2799–2804}
}

@article{Sadhukhan2015,
  title = {Beating no-go theorems by engineering defects in quantum spin models},
  volume = {17},
  ISSN = {1367-2630},
  url = {http://dx.doi.org/10.1088/1367-2630/17/4/043013},
  DOI = {10.1088/1367-2630/17/4/043013},
  number = {4},
  journal = {New Journal of Physics},
  publisher = {IOP Publishing},
  author = {Sadhukhan,  Debasis and Roy,  Sudipto Singha and Rakshit,  Debraj and Sen(De),  Aditi and Sen,  Ujjwal},
  year = {2015},
  month = apr,
  pages = {043013}
}

@article{Horodecki2003,
  title = {Entanglement Breaking Channels},
  volume = {15},
  ISSN = {1793-6659},
  url = {http://dx.doi.org/10.1142/S0129055X03001709},
  DOI = {10.1142/s0129055x03001709},
  number = {06},
  journal = {Reviews in Mathematical Physics},
  publisher = {World Scientific Pub Co Pte Lt},
  author = {Horodecki,  Michael and Shor,  Peter W. and Ruskai,  Mary Beth},
  year = {2003},
  month = aug,
  pages = {629–641}
}

@article{Ruskai2003,
  title = {Qubit Entanglement Breaking Channels},
  volume = {15},
  ISSN = {1793-6659},
  url = {http://dx.doi.org/10.1142/S0129055X03001710},
  DOI = {10.1142/s0129055x03001710},
  number = {06},
  journal = {Reviews in Mathematical Physics},
  publisher = {World Scientific Pub Co Pte Lt},
  author = {Ruskai,  Mary Beth},
  year = {2003},
  month = aug,
  pages = {643–662}
}

@article{error1,
  title = {Decoherence, einselection, and the quantum origins of the classical},
  author = {Zurek, Wojciech Hubert},
  journal = {Rev. Mod. Phys.},
  volume = {75},
  issue = {3},
  pages = {715--775},
  numpages = {0},
  year = {2003},
  month = {May},
  publisher = {American Physical Society},
  doi = {10.1103/RevModPhys.75.715},
  url = {https://link.aps.org/doi/10.1103/RevModPhys.75.715}
}

@article{error2,
  title = {Behavior of Quantum Correlations under Local Noise},
  author = {Streltsov, Alexander and Kampermann, Hermann and Bru\ss{}, Dagmar},
  journal = {Phys. Rev. Lett.},
  volume = {107},
  issue = {17},
  pages = {170502},
  numpages = {5},
  year = {2011},
  month = {Oct},
  publisher = {American Physical Society},
  doi = {10.1103/PhysRevLett.107.170502},
  url = {https://link.aps.org/doi/10.1103/PhysRevLett.107.170502}
}

@article{errormitigation1,
  title = {Quantum error mitigation},
  author = {Cai, Zhenyu and Babbush, Ryan and Benjamin, Simon C. and Endo, Suguru and Huggins, William J. and Li, Ying and McClean, Jarrod R. and O'Brien, Thomas E.},
  journal = {Rev. Mod. Phys.},
  volume = {95},
  issue = {4},
  pages = {045005},
  numpages = {37},
  year = {2023},
  month = {Dec},
  publisher = {American Physical Society},
  doi = {10.1103/RevModPhys.95.045005},
  url = {https://link.aps.org/doi/10.1103/RevModPhys.95.045005}
}

@article{em2,
  title = {Exponentially tighter bounds on limitations of quantum error mitigation},
  volume = {20},
  ISSN = {1745-2481},
  url = {http://dx.doi.org/10.1038/s41567-024-02536-7},
  DOI = {10.1038/s41567-024-02536-7},
  number = {10},
  journal = {Nature Physics},
  publisher = {Springer Science and Business Media LLC},
  author = {Quek,  Yihui and Stilck Fran\c{c}a,  Daniel and Khatri,  Sumeet and Meyer,  Johannes Jakob and Eisert,  Jens},
  year = {2024},
  month = jul,
  pages = {1648–1658}
}

@article{Becher2023,
  title = {2023 roadmap for materials for quantum technologies},
  volume = {3},
  ISSN = {2633-4356},
  url = {http://dx.doi.org/10.1088/2633-4356/aca3f2},
  DOI = {10.1088/2633-4356/aca3f2},
  number = {1},
  journal = {Materials for Quantum Technology},
  publisher = {IOP Publishing},
  author = {Becher,  Christoph and Gao,  Weibo and Kar,  Swastik and Marciniak,  Christian D and Monz,  Thomas and Bartholomew,  John G and Goldner,  Philippe and Loh,  Huanqian and Marcellina,  Elizabeth and Goh,  Kuan Eng Johnson and Koh,  Teck Seng and Weber,  Bent and Mu,  Zhao and Tsai,  Jeng-Yuan and Yan,  Qimin and Huber-Loyola,  Tobias and H\"{o}fling,  Sven and Gyger,  Samuel and Steinhauer,  Stephan and Zwiller,  Val},
  year = {2023},
  month = jan,
  pages = {012501}
}

@article{erasure,
  title = {Capacities of Quantum Erasure Channels},
  author = {Bennett, Charles H. and DiVincenzo, David P. and Smolin, John A.},
  journal = {Phys. Rev. Lett.},
  volume = {78},
  issue = {16},
  pages = {3217--3220},
  numpages = {0},
  year = {1997},
  month = {Apr},
  publisher = {American Physical Society},
  doi = {10.1103/PhysRevLett.78.3217},
  url = {https://link.aps.org/doi/10.1103/PhysRevLett.78.3217}
}

@article{eco,
  title = {Quantum networks boosted by entanglement with a control system},
  author = {Guha, Tamal and Roy, Saptarshi and Chiribella, Giulio},
  journal = {Phys. Rev. Res.},
  volume = {5},
  issue = {3},
  pages = {033214},
  numpages = {16},
  year = {2023},
  month = {Sep},
  publisher = {American Physical Society},
  doi = {10.1103/PhysRevResearch.5.033214},
  url = {https://link.aps.org/doi/10.1103/PhysRevResearch.5.033214}
}

@article{Chiribella2019,
  title = {Quantum Shannon theory with superpositions of trajectories},
  volume = {475},
  ISSN = {1471-2946},
  url = {http://dx.doi.org/10.1098/rspa.2018.0903},
  DOI = {10.1098/rspa.2018.0903},
  number = {2225},
  journal = {Proceedings of the Royal Society A: Mathematical,  Physical and Engineering Sciences},
  publisher = {The Royal Society},
  author = {Chiribella,  Giulio and Kristjánsson,  Hlér},
  year = {2019},
  month = may,
  pages = {20180903}
}

@article{comswitch2,
  title = {Enhanced Communication with the Assistance of Indefinite Causal Order},
  author = {Ebler, Daniel and Salek, Sina and Chiribella, Giulio},
  journal = {Phys. Rev. Lett.},
  volume = {120},
  issue = {12},
  pages = {120502},
  numpages = {5},
  year = {2018},
  month = {Mar},
  publisher = {American Physical Society},
  doi = {10.1103/PhysRevLett.120.120502},
  url = {https://link.aps.org/doi/10.1103/PhysRevLett.120.120502}
}

@article{Chiribella2021,
  title = {Indefinite causal order enables perfect quantum communication with zero capacity channels},
  volume = {23},
  ISSN = {1367-2630},
  url = {http://dx.doi.org/10.1088/1367-2630/abe7a0},
  DOI = {10.1088/1367-2630/abe7a0},
  number = {3},
  journal = {New Journal of Physics},
  publisher = {IOP Publishing},
  author = {Chiribella,  Giulio and Banik,  Manik and Bhattacharya,  Some Sankar and Guha,  Tamal and Alimuddin,  Mir and Roy,  Arup and Saha,  Sutapa and Agrawal,  Sristy and Kar,  Guruprasad},
  year = {2021},
  month = mar,
  pages = {033039}
}

@article{sensing,
  title = {Quantum parameter estimation on coherently superposed noisy channels},
  author = {Chapeau-Blondeau, Fran\ifmmode \mbox{\c{c}}\else \c{c}\fi{}ois},
  journal = {Phys. Rev. A},
  volume = {104},
  issue = {3},
  pages = {032214},
  numpages = {16},
  year = {2021},
  month = {Sep},
  publisher = {American Physical Society},
  doi = {10.1103/PhysRevA.104.032214},
  url = {https://link.aps.org/doi/10.1103/PhysRevA.104.032214}
}

@article{exp1,
  title = {Indefinite Causal Order in a Quantum Switch},
  author = {Goswami, K. and Giarmatzi, C. and Kewming, M. and Costa, F. and Branciard, C. and Romero, J. and White, A. G.},
  journal = {Phys. Rev. Lett.},
  volume = {121},
  issue = {9},
  pages = {090503},
  numpages = {5},
  year = {2018},
  month = {Aug},
  publisher = {American Physical Society},
  doi = {10.1103/PhysRevLett.121.090503},
  url = {https://link.aps.org/doi/10.1103/PhysRevLett.121.090503}
}

@article{exp2,
  title = {Experimental Quantum Switching for Exponentially Superior Quantum Communication Complexity},
  author = {Wei, Kejin and Tischler, Nora and Zhao, Si-Ran and Li, Yu-Huai and Arrazola, Juan Miguel and Liu, Yang and Zhang, Weijun and Li, Hao and You, Lixing and Wang, Zhen and Chen, Yu-Ao and Sanders, Barry C. and Zhang, Qiang and Pryde, Geoff J. and Xu, Feihu and Pan, Jian-Wei},
  journal = {Phys. Rev. Lett.},
  volume = {122},
  issue = {12},
  pages = {120504},
  numpages = {6},
  year = {2019},
  month = {Mar},
  publisher = {American Physical Society},
  doi = {10.1103/PhysRevLett.122.120504},
  url = {https://link.aps.org/doi/10.1103/PhysRevLett.122.120504}
}

@article{exp3,
  title = {Experimental quantum communication enhancement by superposing trajectories},
  author = {Rubino, Giulia and Rozema, Lee A. and Ebler, Daniel and Kristj\'ansson, Hl\'er and Salek, Sina and Allard Gu\'erin, Philippe and Abbott, Alastair A. and Branciard, Cyril and Brukner, \ifmmode \check{C}\else \v{C}\fi{}aslav and Chiribella, Giulio and Walther, Philip},
  journal = {Phys. Rev. Res.},
  volume = {3},
  issue = {1},
  pages = {013093},
  numpages = {19},
  year = {2021},
  month = {Jan},
  publisher = {American Physical Society},
  doi = {10.1103/PhysRevResearch.3.013093},
  url = {https://link.aps.org/doi/10.1103/PhysRevResearch.3.013093}
}

@article{work1,
  title = {Quantum Refrigeration with Indefinite Causal Order},
  author = {Felce, David and Vedral, Vlatko},
  journal = {Phys. Rev. Lett.},
  volume = {125},
  issue = {7},
  pages = {070603},
  numpages = {6},
  year = {2020},
  month = {Aug},
  publisher = {American Physical Society},
  doi = {10.1103/PhysRevLett.125.070603},
  url = {https://link.aps.org/doi/10.1103/PhysRevLett.125.070603}
}

@article{work2,
  title = {Thermodynamic advancement in the causally inseparable occurrence of thermal maps},
  author = {Guha, Tamal and Alimuddin, Mir and Parashar, Preeti},
  journal = {Phys. Rev. A},
  volume = {102},
  issue = {3},
  pages = {032215},
  numpages = {6},
  year = {2020},
  month = {Sep},
  publisher = {American Physical Society},
  doi = {10.1103/PhysRevA.102.032215},
  url = {https://link.aps.org/doi/10.1103/PhysRevA.102.032215}
}

@article{work3,
  title = {Activation of thermal states by coherently controlled thermalization processes},
  volume = {27},
  ISSN = {1367-2630},
  url = {http://dx.doi.org/10.1088/1367-2630/ade5c4},
  DOI = {10.1088/1367-2630/ade5c4},
  number = {7},
  journal = {New Journal of Physics},
  publisher = {IOP Publishing},
  author = {Simonov,  Kyrylo and Roy,  Saptarshi and Guha,  Tamal and Zimborás,  Zoltán and Chiribella,  Giulio},
  year = {2025},
  month = jul,
  pages = {074502}
}

@article{Roy2021,
  title = {Recycling the resource: Sequential usage of shared state in quantum teleportation with weak measurements},
  volume = {392},
  ISSN = {0375-9601},
  url = {http://dx.doi.org/10.1016/j.physleta.2021.127143},
  DOI = {10.1016/j.physleta.2021.127143},
  journal = {Physics Letters A},
  publisher = {Elsevier BV},
  author = {Roy,  Saptarshi and Bera,  Anindita and Mal,  Shiladitya and Sen(De),  Aditi and Sen,  Ujjwal},
  year = {2021},
  month = mar,
  pages = {127143}
}

@article{ub1,
  title = {Multiple Observers Can Share the Nonlocality of Half of an Entangled Pair by Using Optimal Weak Measurements},
  author = {Silva, Ralph and Gisin, Nicolas and Guryanova, Yelena and Popescu, Sandu},
  journal = {Phys. Rev. Lett.},
  volume = {114},
  issue = {25},
  pages = {250401},
  numpages = {5},
  year = {2015},
  month = {Jun},
  publisher = {American Physical Society},
  doi = {10.1103/PhysRevLett.114.250401},
  url = {https://link.aps.org/doi/10.1103/PhysRevLett.114.250401}
}

@article{ub2,
  title = {Sharing of Nonlocality of a Single Member of an Entangled Pair of Qubits Is Not Possible by More than Two Unbiased Observers on the Other Wing},
  volume = {4},
  ISSN = {2227-7390},
  url = {http://dx.doi.org/10.3390/math4030048},
  DOI = {10.3390/math4030048},
  number = {3},
  journal = {Mathematics},
  publisher = {MDPI AG},
  author = {Mal,  Shiladitya and Majumdar,  Archan and Home,  Dipankar},
  year = {2016},
  month = jul,
  pages = {48}
}

@article{ub3,
  title = {Arbitrarily Many Independent Observers Can Share the Nonlocality of a Single Maximally Entangled Qubit Pair},
  author = {Brown, Peter J. and Colbeck, Roger},
  journal = {Phys. Rev. Lett.},
  volume = {125},
  issue = {9},
  pages = {090401},
  numpages = {5},
  year = {2020},
  month = {Aug},
  publisher = {American Physical Society},
  doi = {10.1103/PhysRevLett.125.090401},
  url = {https://link.aps.org/doi/10.1103/PhysRevLett.125.090401}
}

@article{ub4,
  title = {Robustness of higher-dimensional nonlocality against dual noise and sequential measurements},
  author = {Roy, Saptarshi and Kumari, Asmita and Mal, Shiladitya and Sen(De), Aditi},
  journal = {Phys. Rev. A},
  volume = {109},
  issue = {6},
  pages = {062227},
  numpages = {10},
  year = {2024},
  month = {Jun},
  publisher = {American Physical Society},
  doi = {10.1103/PhysRevA.109.062227},
  url = {https://link.aps.org/doi/10.1103/PhysRevA.109.062227}
}

@article{s1,
  title = {Steering a single system sequentially by multiple observers},
  author = {Sasmal, Souradeep and Das, Debarshi and Mal, Shiladitya and Majumdar, A. S.},
  journal = {Phys. Rev. A},
  volume = {98},
  issue = {1},
  pages = {012305},
  numpages = {6},
  year = {2018},
  month = {Jul},
  publisher = {American Physical Society},
  doi = {10.1103/PhysRevA.98.012305},
  url = {https://link.aps.org/doi/10.1103/PhysRevA.98.012305}
}

@article{s2,
  title = {Unbounded sequence of observers exhibiting Einstein-Podolsky-Rosen steering},
  author = {Shenoy H., Akshata and Designolle, S\'ebastien and Hirsch, Flavien and Silva, Ralph and Gisin, Nicolas and Brunner, Nicolas},
  journal = {Phys. Rev. A},
  volume = {99},
  issue = {2},
  pages = {022317},
  numpages = {6},
  year = {2019},
  month = {Feb},
  publisher = {American Physical Society},
  doi = {10.1103/PhysRevA.99.022317},
  url = {https://link.aps.org/doi/10.1103/PhysRevA.99.022317}
}

@article{ue1,
  title = {Witnessing bipartite entanglement sequentially by multiple observers},
  author = {Bera, Anindita and Mal, Shiladitya and Sen(De), Aditi and Sen, Ujjwal},
  journal = {Phys. Rev. A},
  volume = {98},
  issue = {6},
  pages = {062304},
  numpages = {7},
  year = {2018},
  month = {Dec},
  publisher = {American Physical Society},
  doi = {10.1103/PhysRevA.98.062304},
  url = {https://link.aps.org/doi/10.1103/PhysRevA.98.062304}
}

@article{ue2,
  title = {Sequential measurement-device-independent entanglement detection by multiple observers},
  author = {Srivastava, Chirag and Mal, Shiladitya and Sen(De), Aditi and Sen, Ujjwal},
  journal = {Phys. Rev. A},
  volume = {103},
  issue = {3},
  pages = {032408},
  numpages = {9},
  year = {2021},
  month = {Mar},
  publisher = {American Physical Society},
  doi = {10.1103/PhysRevA.103.032408},
  url = {https://link.aps.org/doi/10.1103/PhysRevA.103.032408}
}

@article{ue3,
  title = {Unbounded recycling of genuine multiparty entanglement for any number of qubits},
  author = {Srivastava, Chirag and Pandit, Mahasweta and Sen, Ujjwal},
  journal = {Phys. Rev. A},
  volume = {111},
  issue = {1},
  pages = {012413},
  numpages = {12},
  year = {2025},
  month = {Jan},
  publisher = {American Physical Society},
  doi = {10.1103/PhysRevA.111.012413},
  url = {https://link.aps.org/doi/10.1103/PhysRevA.111.012413}
}

@article{Mermin1990,
  title = {Quantum mysteries revisited},
  volume = {58},
  ISSN = {1943-2909},
  url = {http://dx.doi.org/10.1119/1.16503},
  DOI = {10.1119/1.16503},
  number = {8},
  journal = {American Journal of Physics},
  publisher = {American Association of Physics Teachers (AAPT)},
  author = {Mermin,  N. David},
  year = {1990},
  month = aug,
  pages = {731–734}
}

@article{durvidalcirac,
  title = {Three qubits can be entangled in two inequivalent ways},
  author = {D\"ur, W. and Vidal, G. and Cirac, J. I.},
  journal = {Phys. Rev. A},
  volume = {62},
  issue = {6},
  pages = {062314},
  numpages = {12},
  year = {2000},
  month = {Nov},
  publisher = {American Physical Society},
  doi = {10.1103/PhysRevA.62.062314},
  url = {https://link.aps.org/doi/10.1103/PhysRevA.62.062314}
}

@article{ppt,
  title = {Separability Criterion for Density Matrices},
  author = {Peres, Asher},
  journal = {Phys. Rev. Lett.},
  volume = {77},
  issue = {8},
  pages = {1413--1415},
  numpages = {0},
  year = {1996},
  month = {Aug},
  publisher = {American Physical Society},
  doi = {10.1103/PhysRevLett.77.1413},
  url = {https://link.aps.org/doi/10.1103/PhysRevLett.77.1413}
}

@article{Horodecki1995May,
	author = {Horodecki, R. and Horodecki, P. and Horodecki, M.},
	title = {{Violating Bell inequality by mixed spin-12 states: necessary and sufficient condition}},
	journal = {Phys. Lett. A},
	volume = {200},
	number = {5},
	pages = {340--344},
	year = {1995},
	month = may,
	issn = {0375-9601},
	publisher = {North-Holland},
	doi = {10.1016/0375-9601(95)00214-N}
}

@article{Clauser1969Oct,
	author = {Clauser, John F. and Horne, Michael A. and Shimony, Abner and Holt, Richard A.},
	title = {{Proposed Experiment to Test Local Hidden-Variable Theories}},
	journal = {Phys. Rev. Lett.},
	volume = {23},
	number = {15},
	pages = {880--884},
	year = {1969},
	month = oct,
	publisher = {American Physical Society},
	doi = {10.1103/PhysRevLett.23.880}
}

@article{Cirelson1980Mar,
	author = {Cirel'son, B. S.},
	title = {{Quantum generalizations of Bell's inequality}},
	journal = {Lett. Math. Phys.},
	volume = {4},
	number = {2},
	pages = {93--100},
	year = {1980},
	month = mar,
	issn = {1573-0530},
	publisher = {Kluwer Academic Publishers},
	doi = {10.1007/BF00417500}
}

@article{Clauser1974Jul,
	author = {Clauser, John F. and Horne, Michael A.},
	title = {{Experimental consequences of objective local theories}},
	journal = {Phys. Rev. D},
	volume = {10},
	number = {2},
	pages = {526--535},
	year = {1974},
	month = jul,
	publisher = {American Physical Society},
	doi = {10.1103/PhysRevD.10.526}
}

@article{Brunner2014Apr,
	author = {Brunner, Nicolas and Cavalcanti, Daniel and Pironio, Stefano and Scarani, Valerio and Wehner, Stephanie},
	title = {{Bell nonlocality}},
	journal = {Rev. Mod. Phys.},
	volume = {86},
	number = {2},
	pages = {419--478},
	year = {2014},
	month = apr,
	publisher = {American Physical Society},
	doi = {10.1103/RevModPhys.86.419}
}

@article{Chiribella2013Aug,
	author = {Chiribella, Giulio and D{'}Ariano, Giacomo Mauro and Perinotti, Paolo and Valiron, Benoit},
	title = {{Quantum computations without definite causal structure}},
	journal = {Phys. Rev. A},
	volume = {88},
	number = {2},
	pages = {022318},
	year = {2013},
	month = aug,
	publisher = {American Physical Society},
	doi = {10.1103/PhysRevA.88.022318}
}

@article{Chiribella2009Aug,
	author = {Chiribella, Giulio and D{'}Ariano, Giacomo Mauro and Perinotti, Paolo},
	title = {{Theoretical framework for quantum networks}},
	journal = {Phys. Rev. A},
	volume = {80},
	number = {2},
	pages = {022339},
	year = {2009},
	month = aug,
	publisher = {American Physical Society},
	doi = {10.1103/PhysRevA.80.022339}
}

\appendix
\begin{widetext}

\section{Preliminaries}
\label{sec:preliminaries}

We review some preliminary concepts that form the foundation of our work. We begin by introducing the mathematical formalism of quantum channels to describe noisy dynamics. We then discuss the quantum {\tt SWITCH} as a means to create superpositions of causal orders, followed by an overview of the CHSH inequality for certifying Bell nonlocality. Finally, we detail generalized and unsharp measurements.

\subsection{Noisy Dynamics}
Quantum systems are rarely perfectly isolated from their surroundings. The dynamics of a system interacting with an external environment, often referred to as noise, are described by the theory of open quantum systems. Within this framework, the evolution of the system's density matrix, $\rho$, is modeled by a quantum channel: a linear map $\map E$ that is completely positive and trace-preserving (CPTP).

The action of any such channel can be expressed using the operator-sum representation (or Kraus representation):
\begin{equation}
    \map E(\rho) = \sum_i E_i \rho E_i^\dag,
\end{equation}
where the operators $\{E_i\}$, known as Kraus operators, describe the effect of the channel on the system's Hilbert space. The trace-preserving property of the channel, which ensures that the state remains normalized, imposes the condition $\sum_i E_i^\dag E_i = \mathbb{I}$, where $\mathbb{I}$ is the identity operator. Physically, each term in the sum can be thought of as a possible evolutionary path the quantum system can take due to the environmental interaction.

\subsection{The Quantum {\tt {\tt SWITCH}}}
While quantum channels describe the evolution of a system, higher-order operations known as supermaps can transform these channels themselves. The quantum {\tt SWITCH} is a prominent example of a supermap that takes two quantum channels, $\map E$ and $\map F$, as inputs and arranges them in a new configuration whose causal order is governed by a quantum degree of freedom \cite{Chiribella2009Aug, Chiribella2013Aug}. Specifically, the quantum {\tt SWITCH} generates a new channel, $\mathcal{S}(\map E,\map F)$, that acts on a target system and an ancillary control qubit. The order in which the channels $\map E$ and $\map F$ are applied to the target depends on the state of the control. If the control qubit is in the state $\ket{0}$, the channels are applied in the order $\map F \circ \map E$. If the control is in the state $\ket{1}$, the order is $\map E \circ \map F$. This conditional logic is captured by a single set of Kraus operators $\{S_{ij}\}$ acting on the joint target-control system:
\begin{equation}
    S_{ij} = F_{i} E_{j}\otimes\ketbra{0}{0} + E_{j} F_{i}\otimes\ketbra{1}{1},
\end{equation}
where $\{E_j\}$ and $\{F_i\}$ are the Kraus operators for the channels $\map E$ and $\map F$, respectively. The overall transformation is then given by
\begin{equation}
    \mathcal{S}(\map E, \map F)(\rho \otimes \rho_c) = \sum_{i,j} S_{ij} (\rho \otimes \rho_c) S_{ij}^\dag,
\end{equation}
where $\rho_c$ is the state of the control system. A remarkable feature of the quantum {\tt SWITCH} arises when the control qubit is prepared in a superposition, such as $\ket{+} = (\ket{0}+\ket{1})/\sqrt{2}$. In this case, the target system evolves through a quantum superposition of the two alternative causal orderings, $\map F \circ \map E$ and $\map E \circ \map F$. The {\tt {\tt SWITCH}} supermap possesses an indefinite causal order; that is, it can not be written as a convex combination of quantum combs. This marks a fundamentally non-classical feature that has been shown to provide advantages in various quantum information tasks.

\subsection{Bell Nonlocality and the CHSH Inequality}
A key objective of our work is to certify quantum correlations that defy classical explanation. The primary tool for this is the Bell test, specifically the Clauser-Horn-Shimony-Holt (CHSH) inequality \cite{Clauser1969Oct, Clauser1974Jul}. The CHSH test considers two spatially separated parties, Alice and Bob, who share a quantum state $\rho_{AB}$. Alice can perform one of two measurements, labeled $a$ and $a'$, while Bob can choose between two measurements, $b$ and $b'$. Each measurement has two possible outcomes, denoted $\pm 1$.

The principles of local realism, which assume that measurement outcomes are predetermined by local hidden variables and are independent of spacelike separated events, impose a classical bound on the correlations between their outcomes. These correlations are captured by the Bell expression:
\begin{equation}
 \mathcal{B} = E(a,b) + E(a', b) + E(a, b') - E(a', b'),
\label{eq:bellexp_redef}
\end{equation}
where $E(x,y) = \langle A_x B_y \rangle_\rho$ is the expectation value of the product of Alice's outcome for setting $x$ and Bob's for setting $y$. Local realism dictates that $|\mathcal{B}| \leq 2$. However, quantum mechanics predicts that this bound can be violated. For quantum systems, the maximum possible value is given by the Tsirelson bound, $|\mathcal{B}^Q|_{\max} = 2\sqrt{2}$ \cite{Cirelson1980Mar}. An observed violation of the classical bound, $|\mathcal{B}| > 2$, is a definitive signature of Bell nonlocality.

For any given two-qubit state $\rho$, the maximum possible CHSH violation can be calculated directly using Horodecki's criterion \cite{Horodecki1995May}. This involves constructing a $3 \times 3$ real matrix $T_\rho$ with elements $[T_{\rho}]_{ij} = \Tr[\rho(\sigma_i\otimes \sigma_j)]$ for $i,j \in \{1,2,3\}$, where $\sigma_i$ are the Pauli matrices. From this, one computes the matrix $U_\rho = T_\rho^T T_\rho$. If $u_1$ and $u_2$ are the two largest eigenvalues of $U_\rho$, then the maximal Bell violation for the state $\rho$ is given by $|\mathcal{B}|_{\max} = 2\sqrt{u_1+u_2}$. This provides a direct method to quantify the nonlocality of any two-qubit state.

\subsection{Unsharp Measurements}
The textbook concept of a projective measurement is an idealization. A more general and physically complete description of measurements in quantum mechanics is provided by the Positive Operator-Valued Measure (POVM) formalism. A POVM is a set of operators $\{E_l\}$ that satisfy two conditions: each operator $E_l$ is positive semi-definite ($E_l \ge 0$), and the set forms a resolution of the identity ($\sum_l E_l = \mathbb{I}$).

When a measurement described by the POVM $\{E_l\}$ is performed on a system in state $\rho$, the probability of obtaining outcome $l$ is given by Born's rule:
\begin{equation}
    p(l) = \Tr(\rho E_l).
\end{equation}
An "unsharp measurement" is a type of POVM that smoothly interpolates between a sharp, projective measurement and a trivial measurement that extracts no information. In our protocol, we employ two-outcome unsharp measurements on the control qubit in the basis $\{\ket{+},\ket{-}\}$. The corresponding POVM elements are
\begin{equation}
    E_{\pm}^{\lambda} = \lambda \ketbra{\pm}{\pm} + \frac{1-\lambda}{2} \mathbb I = \frac{1+\lambda}{2} \ketbra{\pm}{\pm} + \frac{1-\lambda}{2}\ketbra{\mp}{\mp},
\end{equation}
where the parameter $\lambda \in [0, 1]$ controls the "sharpness" of the measurement. When $\lambda=1$, this reduces to a standard projective measurement. When $\lambda=0$, the POVM elements become $E_{\pm} = \mathbb{I}/2$, and the measurement provides no information about the state.

Crucially, upon obtaining a specific outcome, the state of the system undergoes a conditional evolution (a process known as state reduction or collapse). Each POVM element can be expressed as $E_l = M_l^\dag M_l$, where $M_l$ is the corresponding measurement operator. If outcome $l$ is obtained, the state is updated according to the selective update rule:
\begin{equation}
    \rho \xrightarrow{\text{outcome } l} \rho'_l = \frac{M_l \rho M_l^\dag}{p(l)} = \frac{M_l \rho M_l^\dag}{\Tr(M_l \rho M_l^\dag)}.
\end{equation}
For the unsharp measurements in our protocol, the measurement operators can be chosen as $M_\pm = \sqrt{E_\pm^\lambda}$. This rule for updating the state conditioned on the measurement outcome is essential for analyzing the step-by-step evolution in our sequential protocol.

\section{Proof of Theorem \ref{Th:rhocab_general}}
\label{app:proof_th1}
The base cases for the induction, $k=1,2$, were verified in the main text. Here, we prove the inductive step of Theorem~\ref{Th:rhocab_general}. Assuming that our hypothesis is true for some $k\geq1$, thus

\begin{eqnarray}
\rho_{CAB_{k}} &=& \Bigg[\frac{1+R(k-1)}{4} +\frac{S(k-1)}{2}\Bigg]\ket{\text{GHZ}_{\alpha+}}\bra{\text{GHZ}_{\alpha+}} + \Bigg[\frac{1+R(k-1)}{4} -\frac{S(k-1)}{2}\Bigg]\ket{\text{GHZ}_{\alpha-}}\bra{\text{GHZ}_{\alpha-}} \notag\\
&+& \Bigg[\frac{1-R(k-1)}{4}\Bigg] \Bigg(2(1-\alpha)\ket{010}\bra{010}+ 2   \alpha\ket{101}\bra{101}\Bigg) 
\end{eqnarray}
We proceed with the protocol to derive the state $\rho_{CAB_{k+1}}$ and show that it indeed satisfies the above form for $k+1$. Now, $\rho_{CAB_{k}}$ undergoes the unsharp measurement described by POVMs in Eq.~\eqref{eqn:unsharp_measure} with measurement operators
$\sqrt{E_\pm^{\lambda_k}} = \sqrt{\frac{1\pm \lambda_k}{2}}\ket{+}\bra{+} + \sqrt{\frac{1\mp \lambda_k}{2}}\ket{-}\bra{-}$. These produce the general post-measurement states 

\begin{eqnarray}
\rho_{CAB_{k}}^{E_\pm} &=& \Bigg[\frac{1+R(k-1)}{4} +\frac{S(k-1)}{2}\Bigg]\Bigg(\frac{1\pm \lambda_k}{2} \ket{+\Phi^+_\alpha}\bra{+\Phi^+_\alpha}+\frac{1\mp \lambda_k}{2} \ket{-\Phi^-_\alpha}\bra{-\Phi^-_\alpha} \notag\\ && \hspace{6cm}+\frac{\sqrt{1- \lambda_k^2}}{2} \ket{+\Phi^+_\alpha}\bra{-\Phi^-_\alpha}+\frac{\sqrt{1- \lambda_k^2}}{2} \ket{-\Phi^-_\alpha}\bra{+\Phi^+_\alpha}\Bigg) \notag\\
&&+ \Bigg[\frac{1+R(k-1)}{4} -\frac{S(k-1)}{2}\Bigg]\Bigg(\frac{1\pm \lambda_k}{2} \ket{+\Phi^-_\alpha}\bra{+\Phi^-_\alpha}+\frac{1\mp \lambda_k}{2} \ket{-\Phi^+_\alpha}\bra{-\Phi^+_\alpha}  \notag\\ && \hspace{6cm}+\frac{\sqrt{1- \lambda_k^2}}{2} \ket{+\Phi^-_\alpha}\bra{-\Phi^+_\alpha}+\frac{\sqrt{1- \lambda_k^2}}{2} \ket{-\Phi^+_\alpha}\bra{+\Phi^-_\alpha}\Bigg) \notag\\
&&+ \Bigg[\frac{1-R(k-1)}{4}\Bigg] \Bigg(\frac{1\pm \lambda_k}{2}\Bigg\{\ket{+\Psi^+_\alpha}\bra{+\Psi^+_\alpha} + \ket{+\Psi^-_\alpha}\bra{+\Psi^-_\alpha} \Bigg\} \notag\\ && \hspace{1cm}+ \frac{1\mp \lambda_k}{2}\Bigg\{\ket{-\Psi^+_\alpha}\bra{-\Psi^+_\alpha} + \ket{-\Psi^-_\alpha}\bra{-\Psi^-_\alpha} \Bigg\} \notag\\ && \hspace{2cm}+ \frac{\sqrt{1- \lambda_k^2}}{2}\Bigg\{(1-\alpha)\ket{010}\bra{010}+\alpha\ket{101}\bra{101}-(1-\alpha)\ket{110}\bra{110}-\alpha\ket{001}\bra{001}\Bigg\}\Bigg).
\label{eq:appendix_postmeasurement_state}
\end{eqnarray}

 With this state at hand the referee R can ask to deduce Bell violation in the A:B cut. Alternatively, the state can be sent to the next round of protocol where the {\tt {\tt SWITCH}} action from 
Eq.~\eqref{eqn:Switch_Kraus} takes place to produce the desired state $\rho_{CAB_{k+1}}$. 

\begin{eqnarray}
&&\rho_{CAB_{k+1}} = \notag\\
&&\Bigg[\frac{1+R(k-1)}{4} +\frac{S(k-1)}{2}\Bigg]\Bigg(\frac{1+ \sqrt{1-\lambda_k^2}}{2} \ket{\text{GHZ}_{\alpha+}}\bra{\text{GHZ}_{\alpha+}}+\frac{1- \sqrt{1-\lambda_k^2}}{4} \Bigg\{2(1-\alpha)\ket{010}\bra{010} + 2\alpha\ket{101}\bra{101}\Bigg\}\Bigg)  \notag\\
&&+ \Bigg[\frac{1+R(k-1)}{4} -\frac{S(k-1)}{2}\Bigg]\Bigg(\frac{1+ \sqrt{1-\lambda_k^2}}{2} \ket{\text{GHZ}_{\alpha-}}\bra{\text{GHZ}_{\alpha-}}+\frac{1- \sqrt{1-\lambda_k^2}}{4} \Bigg\{2(1-\alpha)\ket{010}\bra{010} + 2\alpha\ket{101}\bra{101}\Bigg\}\Bigg) \notag\\
&&+ \Bigg[\frac{1-R(k-1)}{4}\Bigg] \Bigg(\alpha\ket{000}\bra{000}+(1-\alpha)\ket{010}\bra{010}+\alpha\ket{101}\bra{101}+(1-\alpha)\ket{111}\bra{111}\Bigg)  \notag\\ && + \Bigg[\frac{1-R(k-1)}{4}\sqrt{1- \lambda_k^2}\Bigg]\Bigg(-\alpha\ket{000}\bra{000}+(1-\alpha)\ket{010}\bra{010}+\alpha\ket{101}\bra{101}-(1-\alpha)\ket{111}\bra{111}\Bigg).
\end{eqnarray}

Now, notice that we can rewrite $2\alpha\ket{000}\bra{000} + 2(1-\alpha)\ket{111}\bra{111} = \ket{\text{GHZ}_{\alpha+}}\bra{\text{GHZ}_{\alpha+}} + \ket{\text{GHZ}_{\alpha-}}\bra{\text{GHZ}_{\alpha-}}$ to get

\begin{eqnarray}
&&\rho_{CAB_{k+1}} = \notag\\
&&\Bigg[\frac{1+R(k-1)}{4} +\frac{S(k-1)}{2}\Bigg]\Bigg(\frac{1+ \sqrt{1-\lambda_k^2}}{2} \ket{\text{GHZ}_{\alpha+}}\bra{\text{GHZ}_{\alpha+}}+\frac{1- \sqrt{1-\lambda_k^2}}{4} \Bigg\{2(1-\alpha)\ket{010}\bra{010} + 2\alpha\ket{101}\bra{101}\Bigg\}\Bigg)  \notag\\
&&+ \Bigg[\frac{1+R(k-1)}{4} -\frac{S(k-1)}{2}\Bigg]\Bigg(\frac{1+ \sqrt{1-\lambda_k^2}}{2} \ket{\text{GHZ}_{\alpha-}}\bra{\text{GHZ}_{\alpha-}}+\frac{1- \sqrt{1-\lambda_k^2}}{4} \Bigg\{2(1-\alpha)\ket{010}\bra{010} + 2\alpha\ket{101}\bra{101}\Bigg\}\Bigg) \notag\\
&&+ \Bigg[\frac{1-R(k-1)}{4}\Bigg] \Bigg(\frac{1}{2}\ket{\text{GHZ}_\alpha+}\bra{\text{GHZ}_{\alpha+}}+(1-\alpha)\ket{010}\bra{010}+\alpha\ket{101}\bra{101}+\frac{1}{2}\ket{\text{GHZ}_{\alpha-}}\bra{\text{GHZ}_{\alpha-}}\Bigg)  \notag\\ && + \Bigg[\frac{1-R(k-1)}{4}\sqrt{1- \lambda_k^2}\Bigg]\Bigg(-\frac{1}{2}\ket{\text{GHZ}_{\alpha+}}\bra{\text{GHZ}_{\alpha+}}+(1-\alpha)\ket{010}\bra{010}+\alpha\ket{101}\bra{101}-\frac{1}{2}\ket{\text{GHZ}_{\alpha-}}\bra{\text{GHZ}_{\alpha-}}\Bigg). \notag\\
\end{eqnarray}
Collecting similar terms together

\begin{eqnarray}
\rho_{CAB_{k+1}} &=&\Bigg(\Bigg[\frac{1+R(k-1)}{4} +\frac{S(k-1)}{2}\Bigg]\frac{1+\sqrt{1-\lambda_k^2}}{2} +\frac{1-R(k-1)}{4} \Bigg[\frac{1}{2}-\frac{\sqrt{1-\lambda_k^2}}{2}\Bigg]\Bigg)\ket{\text{GHZ}_{\alpha+}}\bra{\text{GHZ}_{\alpha+}} \notag\\
&+&\Bigg(\Bigg[\frac{1+R(k-1)}{4} -\frac{S(k-1)}{2}\Bigg]\frac{1+\sqrt{1-\lambda_k^2}}{2} +\frac{1-R(k-1)}{4} \Bigg[\frac{1}{2}-\frac{\sqrt{1-\lambda_k^2}}{2}\Bigg]\Bigg)\ket{\text{GHZ}_{\alpha-}}\bra{\text{GHZ}_{\alpha-}} \notag\\
&+&\Bigg(\frac{1+R(k-1)}{2}\frac{1-\sqrt{1-\lambda_k^2}}{4} + \frac{1-R(k-1)}{4}\Bigg[\frac{1}{2}+\frac{\sqrt{1-\lambda_k^2}}{2}\Bigg]\Bigg)\left(2(1-\alpha)\ket{010}\bra{010} + 2\alpha\ket{101}\bra{101}\right),\notag\\
\end{eqnarray}
and simplifying the terms in the bracket
\begin{eqnarray}
\Bigg[\frac{1+R(k-1)}{4} +\frac{S(k-1)}{2}\Bigg]\frac{1+\sqrt{1-\lambda_k^2}}{2} +\frac{1-R(k-1)}{4} \Bigg[\frac{1}{2}-\frac{\sqrt{1-\lambda_k^2}}{2}\Bigg] &=& \frac{S(k)}{2} + \frac{1}{8} +\frac{R(k-1)}{8} + \frac{\sqrt{1-\lambda_k^2}}{8} + \frac{R(k)}{8} \notag\\
&& + \frac{1}{8} -\frac{R(k-1)}{8} - \frac{\sqrt{1-\lambda_k^2}}{8} + \frac{R(k)}{8} \notag\\
&=& \frac{1+R(k)}{4} + \frac{S(k)}{2}, \notag\\
\Bigg[\frac{1+R(k-1)}{4} -\frac{S(k-1)}{2}\Bigg]\frac{1+\sqrt{1-\lambda_k^2}}{2} +\frac{1-R(k-1)}{4} \Bigg[\frac{1}{2}-\frac{\sqrt{1-\lambda_k^2}}{2}\Bigg] &=& -\frac{S(k)}{2} + \frac{1}{8} +\frac{R(k-1)}{8} + \frac{\sqrt{1-\lambda_k^2}}{8} + \frac{R(k)}{8} \notag\\
&& + \frac{1}{8} -\frac{R(k-1)}{8} - \frac{\sqrt{1-\lambda_k^2}}{8} + \frac{R(k)}{8} \notag\\
&=& \frac{1+R(k)}{4} - \frac{S(k)}{2}, \notag\\
\frac{1+R(k-1)}{2}\frac{1-\sqrt{1-\lambda_k^2}}{4} + \frac{1-R(k-1)}{4}\Bigg[\frac{1}{2}+\frac{\sqrt{1-\lambda_k^2}}{2}\Bigg] &=& \frac{1}{8} +\frac{R(k-1)}{8} - \frac{\sqrt{1-\lambda_k^2}}{8} - \frac{R(k)}{8} \notag\\
&& + \frac{1}{8} -\frac{R(k-1)}{8} + \frac{\sqrt{1-\lambda_k^2}}{8} - \frac{R(k)}{8} \notag\\
&=& \frac{1-R(k)}{4}. 
\end{eqnarray}
We get the final outcome for the state in the $k+1$th round 

\begin{eqnarray}
\rho_{CAB_{k+1}} &=& \Bigg[\frac{1+R(k)}{4} +\frac{S(k)}{2}\Bigg]\ket{\text{GHZ}_{\alpha+}}\bra{\text{GHZ}_{\alpha+}} + \Bigg[\frac{1+R(k)}{4} -\frac{S(k)}{2}\Bigg]\ket{\text{GHZ}_{\alpha-}}\bra{\text{GHZ}_{\alpha-}} \notag\\
&+& \Bigg[\frac{1-R(k)}{4}\Bigg] \Bigg(2(1-\alpha)\ket{010}\bra{010}+2\alpha \ket{101}\bra{101}\Bigg),
\end{eqnarray}

thus proving the induction step. 

\section{Proof of Theorem \ref{Th:Tmatrix}}
\label{app:proof_th2}

We begin with the post-measurement quantum states of the AB subsystem during the $k$th round of the protocol 

\begin{eqnarray}
\rho_{AB_{k}}^{E_\pm}&=& \Bigg[\frac{1+R(k-1)}{4} +\frac{S(k-1)}{2}\Bigg]\Bigg(\frac{1\pm \lambda_k}{2} \ket{\Phi^+_\alpha}\bra{\Phi^+_\alpha}+\frac{1\mp \lambda_k}{2} \ket{\Phi^-_\alpha}\bra{\Phi^-_\alpha} \Bigg) \notag\\
&&+ \Bigg[\frac{1+R(k-1)}{4} -\frac{S(k-1)}{2}\Bigg]\Bigg(\frac{1\pm \lambda_k}{2} \ket{\Phi^-_\alpha}\bra{\Phi^-_\alpha}  +\frac{1\mp \lambda_k}{2} \ket{\Phi^+_\alpha}\bra{\Phi^+_\alpha} \Bigg) \notag \\
&&+ \Bigg[\frac{1-R(k-1)}{4}\Bigg] \Bigg(\ket{\Psi^+_\alpha}\bra{\Psi^+_\alpha} + \ket{\Psi^-_\alpha}\bra{\Psi^-_\alpha}\Bigg).
\end{eqnarray}

With $T_\rho$ matrix being linear on quantum states and the state being in the Bell basis the matrix can be calculated as

\begin{eqnarray}
&&T_{\rho_{AB_{k}}^{E_\pm}}  \notag\\
&&=\Bigg[\frac{1+R(k-1)}{4} +\frac{S(k-1)}{2}\Bigg]\Bigg(\frac{1\pm \lambda_k}{2} \text{Diag}( 2\sqrt{\alpha(1-\alpha)},- 2\sqrt{\alpha(1-\alpha)},1)+\frac{1\mp \lambda_k}{2} \text{Diag}(- 2\sqrt{\alpha(1-\alpha)}, 2\sqrt{\alpha(1-\alpha)},1) \Bigg) \notag\\
&&+ \Bigg[\frac{1+R(k-1)}{4} -\frac{S(k-1)}{2}\Bigg]\Bigg(\frac{1\pm \lambda_k}{2} \text{Diag}(- 2\sqrt{\alpha(1-\alpha)}, 2\sqrt{\alpha(1-\alpha)},1) +\frac{1\mp \lambda_k}{2}\text{Diag}( 2\sqrt{\alpha(1-\alpha)},- 2\sqrt{\alpha(1-\alpha)},1) \Bigg) \notag \\
&&+ \Bigg[\frac{1-R(k-1)}{4}\Bigg] \Bigg(\text{Diag}( 2\sqrt{\alpha(1-\alpha)}, 2\sqrt{\alpha(1-\alpha)},-1) + \text{Diag}(- 2\sqrt{\alpha(1-\alpha)},- 2\sqrt{\alpha(1-\alpha)},-1)\Bigg) \notag\\
&&= \Bigg[\frac{1+R(k-1)}{4} +\frac{S(k-1)}{2}\Bigg]\text{Diag}(\pm  2\sqrt{\alpha(1-\alpha)}\lambda_k, \mp  2\sqrt{\alpha(1-\alpha)}\lambda_k,1) \notag\\
&&+ \Bigg[\frac{1+R(k-1)}{4} -\frac{S(k-1)}{2}\Bigg]\text{Diag}(\mp  2\sqrt{\alpha(1-\alpha)}\lambda_k, \pm  2\sqrt{\alpha(1-\alpha)}\lambda_k,1) + \Bigg[\frac{1-R(k-1)}{4}\Bigg] \text{Diag}(0,0,-2) \notag\\
&&= \text{Diag}\Big(\pm  2\sqrt{\alpha(1-\alpha)}S(k-1)\lambda_k, ~ \mp 2\sqrt{\alpha(1-\alpha)} S(k-1)\lambda_k, ~R(k-1)\Big)
\end{eqnarray}

\section{Proof of Theorem \ref{Th:theorem_Bellviolation} }
\label{app:proof_th3}
For any scaling factor $q$, with which we can set our sharpness parameters to be $\{\lambda_k = 1, ~\lambda_{k-1} = 1/q, \dots, \lambda_m = 1/q^{k-m},\dots,\lambda_1 = 1/q^{k-1}\}$. We wish to show that
\begin{eqnarray}
4\alpha(1-\alpha)\Bigg[\prod_{i=1}^{m-1}\frac{1+\sqrt{1-1/q^{2(k-i)}}}{2}\Bigg]^2 \frac{1}{q^{2(k-m)}} + \Bigg[\prod_{i=1}^{m-1}1-1/q^{2(k-i)}\Bigg]\geq1.
\end{eqnarray}

Now, let $k>m\ge1$ be integers and define
\begin{eqnarray}
X_{m,q} \;=\;4\alpha    (1-\alpha)\Biggl[\prod_{i=1}^{m-1}\frac{1+\sqrt{1-q^{-2(k-i)}}}{2}\Biggr]^{2},
\qquad
Y_{m,q}\;=\;\prod_{i=1}^{m-1}\bigl(1-q^{-2(k-i)}\bigr).
\end{eqnarray}
We wish to show that $X_{m,q}\,q^{-2(k-m)} \;+\;Y_{m,q} \;\ge\;1,$ holds by induction on $m$.

\medskip

\noindent\textbf{Base case ($m=1$).}
When $m=1$, both products are empty, so $X_{1,q}=4\alpha(1-\alpha), ~~Y_{1,q}=1$.  Thus
\begin{eqnarray}
X_{1,q}\,q^{-2(k-1)} + Y_{1,q}
= 1\cdot 4\alpha(1-\alpha) q^{-2(k-1)} + 1
= 1 + 4\alpha(1-\alpha)q^{-2(k-1)}
\ge 1.
\end{eqnarray}
Hence the claim holds for $m=1$.

\medskip

\noindent\textbf{Inductive step.}
Assume the statement holds for some $1\le m\le k$:
\begin{eqnarray}
X_{m,q}\,q^{-2(k-m)} + Y_{m,q} \;\ge\;1.
\end{eqnarray}
We must show it for $m+1$. Setting $ t \;=\;q^{-2(k-m)}, \quad s \;=\;\sqrt{1 - t}.$ Then by definition
\begin{eqnarray}
X_{m+1\,q}
= X_{m,q}\Bigl(\tfrac{1+s}{2}\Bigr)^2,
\quad
Y_{m+1}=Y_m(1-t),
\quad
q^{-2(k-(m+1))}=q^2t.
\end{eqnarray}
Since
\begin{eqnarray}
\Bigl(\tfrac{1+s}{2}\Bigr)^2
=\frac{1+2s+s^2}{4} = \frac{2+2s-t}{4}=\frac{1+s}{2}-\frac{t}{4},
\end{eqnarray}
we obtain
\begin{eqnarray}
X_{m+1,q}\,q^{-2(k-(m+1))}
= X_m\bigl[q^2(1+s)t/2 - q^2t^2/4\bigr].
\end{eqnarray}
Hence
\begin{eqnarray}
\mathrm{LHS}_{m+1,q}
&=& X_{m+1,q}\,q^{-2(k-(m+1))} + Y_{m+1,q} \notag\\
&=& X_{m,q}\bigl[q^2(1+s)t/2 - q^2t^2/4\bigr] + Y_{m,q}(1-t) \notag\\
&=& X_{m,q}\bigl[t + t\bigl(q^2/2+q^2s/2-q^2t/4-1\bigr)\bigr] + Y_{m,q}(1-t) \notag\\
&=& \underbrace{\bigl[X_{m,q}\,t + Y_{m,q}\bigr]}_{\ge1\text{ by IH}}
  \;+\;
  t\bigl[X_{m,q}\,(q^2/2+q^2s/2-q^2t/4-1)\;-\;Y_{m,q}\bigr].
\end{eqnarray}
Call the second bracket $\Delta$.  Then $\Delta
= t\bigl[X_{m,q}\,(q^2/2+q^2s/2-q^2t/4-1)\;-\;Y_{m,q}\bigr].$ Because $t>0$, it suffices to show $
4\alpha(1-\alpha)\Biggl[\prod_{i=1}^{m-1}\frac{1+\sqrt{1-q^{-2(k-i)}}}{2}\Biggr]^{2}\,(q^2/2+q^2s/2-q^2t/4-1)\;\geq\;\prod_{i=1}^{m-1}\bigl(1-q^{-2(k-i)}\bigr)$. We note that
\begin{eqnarray}
  &&4\alpha(1-\alpha)(q^2/2+q^2s/2-q^2t/4-1) \geq1 \notag\\
  &\iff& q^2(2+2s-t)\alpha(1-\alpha) \geq 2 \notag\\
  &\iff&q^2 \geq \frac{2}{(2+2s-t)\alpha(1-\alpha)}
  \end{eqnarray}
  
  Since $ s = \sqrt{1 - t}
    \quad\Longrightarrow\quad
    2+2s-t \ge 1$
  Indeed $ t = q^{-2(k-m)} \le 1$ and at $t = 1,$ $2+2s-t = 1$. Thus setting
  \begin{eqnarray}
  q  &\geq& \sqrt{\frac{2}{a(1-\alpha)}} \; \{0<a<1\} \notag\\
  \text{we get } q^2 &\geq& \frac{2}{a(1-\alpha)} \ge\frac{2}{(2+2s-t)\alpha(1-\alpha)}
  \end{eqnarray}
Furthermore, as each factor satisfies
  \begin{eqnarray}
    \Bigl(\tfrac{1 + \sqrt{1 - q^{-2(k-i)}}}{2}\Bigr)
    \;\ge\;
    \sqrt{1 - q^{-2(k-i)}},
  \end{eqnarray}
  hence $ \Biggl[\prod_{i=1}^{m-1}\frac{1+\sqrt{1-q^{-2(k-i)}}}{2}\Biggr]^{2}\;\ge\; \prod_{i=1}^{m-1}\bigl(1-q^{-2(k-i)}\bigr).$
Therefore
\begin{eqnarray}
X_{m,q}(q^2/2+q^2s/2-q^2t/4-1)-Y_{m,q} \geq0,
\end{eqnarray}
so 
\begin{eqnarray}
\Delta = t\bigl[X_{m,q}\,(q^2/2+q^2s/2-q^2t/4-1)\;-\;Y_{m,q}\bigr]\;\ge\;0.
\end{eqnarray}
Combining,
\begin{eqnarray}
\mathrm{LHS}_{m+1,q}
\;=\;
\bigl[X_{m,q}\,t + Y_{m,q}\bigr] + \Delta
\;\ge\;1 + 0
\;=\;1,
\end{eqnarray}
which completes the inductive step.

\section{Failure of W state for round $k$}
\label{app:proof_th5}

We prove by induction that when the initial resource is a W state, the state $\rho_{CAB_{k}}$ obtained after running the protocol for $k-1$ rounds and implementing the {\tt {\tt SWITCH}} action in the $k$-th round has the form

\begin{eqnarray}
\rho_{CAB_{k}} &=& \frac{1}{3}(\ket{000}\bra{000}+\ket{101}\bra{101}) +\left[\frac{1+R(k-1)}{6}\right]\ket{010}\bra{010} +\left[\frac{1-R(k-1)}{6}\right]\ket{111}\bra{111}
\end{eqnarray}
Let us first verify if the induction hypothesis is true for $k=1$. Starting with the W state. The state after the {\tt {\tt SWITCH}} action is 
$$\rho_{CAB_1} = \mathcal{K}(\ket{W}\bra{W}) = \frac{1}{3}\left(\ket{000}\bra{000}+\ket{010}\bra{010}+\ket{101}\bra{101}\right).$$
 Thus, the induction hypothesis is true at $k=1$. Assuming our hypothesis is true for some $k\geq 1$, we prove the inductive step by deriving the form of $\rho_{CAB_{k+1}}$. For this we perform the unsharp measurement Eq.~\eqref{eqn:unsharp_measure} with measurement operators
$\sqrt{E_\pm^{\lambda_k}} = \sqrt{\frac{1\pm \lambda_k}{2}}\ket{+}\bra{+} + \sqrt{\frac{1\mp \lambda_k}{2}}\ket{-}\bra{-}$ in the control qubit C, to obtain post measurement state

\begin{eqnarray}
\rho_{CAB_{k}}^{E_\pm} &=& \frac{1}{6}\left((1\pm \lambda_k)\ket{+00}\bra{+00}+(1\mp \lambda_k)\ket{-00}\bra{-00}+\sqrt{1-\lambda_k^2}\ket{+00}\bra{-00}+\sqrt{1-\lambda_k^2}\ket{-00}\bra{+00}\right)
\notag\\ && +\frac{1}{6}\left((1\pm \lambda_k)\ket{+01}\bra{+01}+(1\mp \lambda_k)\ket{-01}\bra{-01}-\sqrt{1-\lambda_k^2}\ket{+01}\bra{-01}-\sqrt{1-\lambda_k^2}\ket{-01}\bra{+01}\right) \notag\\
&&+ \left[\frac{1+R(k-1)}{6}\right]\left(\frac{1\pm \lambda_k}{2}\ket{+10}\bra{+10}+\frac{1\mp \lambda_k}{2}\ket{-10}\bra{-10}+\frac{\sqrt{1-\lambda_k^2}}{2}\ket{+10}\bra{-10}+\frac{\sqrt{1-\lambda_k^2}}{2}\ket{-10}\bra{+10}\right)  \notag\\ &&  +\left[\frac{1-R(k-1)}{6}\right]\left(\frac{1\pm \lambda_k}{2}\ket{+11}\bra{+11}+\frac{1\mp \lambda_k}{2}\ket{-11}\bra{-11}-\frac{\sqrt{1-\lambda_k^2}}{2}\ket{+11}\bra{-11}-\frac{\sqrt{1-\lambda_k^2}}{2}\ket{-11}\bra{+11}\right).\notag\\
\label{eq:appendix_postmeasurement_state_W}
\end{eqnarray}

This state is then sent for the next round of the protocol where the {\tt {\tt SWITCH}} action takes place to obtain $\rho_{CAB_{k+1}}$.

\begin{eqnarray}
\rho_{CAB_{k+1}} &=& \frac{1+\sqrt{1-\lambda_k^2}}{6}\ket{000}\bra{000} + \frac{1-\sqrt{1-\lambda_k^2}}{6}\ket{101}\bra{101} \notag\\
&&+\frac{1-\sqrt{1-\lambda_k^2}}{6}\ket{000}\bra{000} + \frac{1+\sqrt{1-\lambda_k^2}}{6}\ket{101}\bra{101} \notag\\
&&+\left[\frac{1+R(k-1)}{6}\right]\left(\frac{1+\sqrt{1-\lambda_k^2}}{2}\ket{010}\bra{010} + \frac{1-\sqrt{1-\lambda_k^2}}{2}\ket{111}\bra{111} \right)\notag\\
&&+\left[\frac{1-R(k-1)}{6}\right]\left(\frac{1-\sqrt{1-\lambda_k^2}}{2}\ket{010}\bra{010} + \frac{1+\sqrt{1-\lambda_k^2}}{2}\ket{111}\bra{111}\right).
\end{eqnarray}

Collecting similar terms together

\begin{eqnarray}
\rho_{CAB_{k+1}} &=& \left[\frac{1+\sqrt{1-\lambda_k^2}}{6}+\frac{1-\sqrt{1-\lambda_k^2}}{6}\right]\ket{000}\bra{000} \notag\\
&&+\left[\frac{1+\sqrt{1-\lambda_k^2}}{6}+\frac{1-\sqrt{1-\lambda_k^2}}{6}\right]\ket{101}\bra{101} \notag\\
&&+\left[\frac{1+R(k-1)}{6}\frac{1+\sqrt{1-\lambda_k^2}}{2}+\frac{1-R(k-1)}{6}\frac{1-\sqrt{1-\lambda_k^2}}{2}\right]\ket{010}\bra{010}  \notag\\
&&+\left[\frac{1+R(k-1)}{6}\frac{1-\sqrt{1-\lambda_k^2}}{2}+\frac{1-R(k-1)}{6}\frac{1+\sqrt{1-\lambda_k^2}}{2}\right]\ket{111}\bra{111}.
\end{eqnarray}

which is simplified to 

\begin{eqnarray}
\rho_{CAB_{k+1}} &=& \frac{1}{3}(\ket{000}\bra{000}+\ket{101}\bra{101}) +\left[\frac{1+R(k)}{6}\right]\ket{010}\bra{010} +\left[\frac{1-R(k)}{6}\right]\ket{111}\bra{111},
\end{eqnarray}

and thus the induction step is proved. Much like the generalized GHZ state, the state is independent of whichever post-measurement state was selected in all the previous rounds.  Performing an unsharp measurement on the control qubit C yields the post-measurement state Eq.~\eqref{eq:appendix_postmeasurement_state_W}. The control of this post-measurement state is then traced out to obtain the joint state on subsystems AB: 
    \begin{eqnarray}
\rho_{AB_{k}}^{E_\pm} &=& \frac{1}{3}\left(\ket{00}\bra{00} + \ket{01}\bra{01}\right)+ \left[\frac{1+R(k-1)}{6}\right]\ket{10}\bra{10} +\left[\frac{1-R(k-1)}{6}\right] \ket{11}\bra{11}.\notag\\
\end{eqnarray}
This state is a separable state with respect to the A:B partition and therefore cannot violate any Bell inequality. We conclude that the W state cannot provide any Bell violation in any round, and thus $k_{\max} = 0$.
\end{widetext}
\end{document}